\documentclass[12pt]{article}

\usepackage[T1]{fontenc}
\usepackage{lmodern}
\usepackage[margin=1in]{geometry}
\usepackage{microtype}
\usepackage{amsmath, amssymb, amsthm, mathtools, bm}

\usepackage{booktabs, array, longtable}
\usepackage{graphicx}
\usepackage{lscape}
\usepackage{float}
\usepackage{multirow}
\usepackage{enumitem, xcolor, tikz}
\usepackage[authoryear,round]{natbib}
\usepackage[
  colorlinks=true,
  linkcolor=blue,
  citecolor=blue,
  urlcolor=blue
]{hyperref}
\usepackage{url}
\usepackage{indentfirst}
\usepackage{setspace}
\newenvironment{appendices}{}{}
\newcommand{\E}{\mathbb{E}}

\newcommand{\var}{\mathrm{Var}}

\newcommand{\indicator}{\mathbb{I}}
\newcommand{\Pn}{\mathbb{P}_n}
\newcommand{\Gn}{\mathbb{G}_n}

\newcommand{\T}{\mathrm{T}}
\newcommand*\diff{\mathop{}\!\mathrm{d}}
\newcommand{\indep}{\perp\!\!\!\perp}

\definecolor{ForestGreen}{rgb}{0.0, 0.5, 0.0}

\newcommand{\expit}{\text{expit}}
\newcommand{\logit}{\text{logit}}
\allowdisplaybreaks
\newtheorem{assumption}{Assumption}
\newtheorem{theorem}{Theorem}
\newtheorem{proposition}{Proposition}
\newtheorem{lemma}{Lemma}
\newtheorem{definition}{Definition}
\newtheorem{remark}{Remark}

\newcommand{\oracle}{o}
\newcommand{\vc}{\mathrm{VC}}
\newcommand{\no}{\mathrm{no}}
\newcommand{\EP}{\mathbb{E}_{\mathbb{P}}}
\newcommand{\prob}{\mathbb{P}}

\title{\bf\LARGE Orthogonal Policy Learning with \\ Ordinal Outcomes}
\author{\bf\small Yue Zhang$^1$\thanks{Email: \href{mailto:zhang_yue@stu.pku.edu.cn}{zhang\_yue@stu.pku.edu.cn}},
Shanshan Luo$^2$\thanks{Corresponding author. Email: \href{mailto:shanshanluo@btbu.edu.cn}{shanshanluo@btbu.edu.cn}}, Yangbo He$^1$\thanks{Corresponding author.
Email: \href{mailto:heyb@math.pku.edu.cn}{heyb@math.pku.edu.cn}}\\
\small $^1$School of Mathematical Sciences, Peking University
\\\small $^2$School of Mathematics and Statistics, Beijing Technology and Business University}

\date{}
\hypersetup{pdftitle={Orthogonal Policy Learning with Ordinal Outcomes},
  pdfauthor={Yue Zhang, Shanshan Luo, Yangbo He}}

\begin{document}

\maketitle

\begin{abstract}
Policy learning methods based on conditional average treatment effects can obscure subpopulation heterogeneity when applied to ordinal outcomes. We develop a policy learning framework for ordinal outcomes with heterogeneous utilities for individuals who strictly benefit from treatment and those who do not. Since the probability of strict benefit is only partially identified without further conditions, we adopt a minimax strategy which minimizes the worst-case regret over the identification region. To estimate the resulting nonsmooth objective, we combine smooth approximations with Neyman orthogonalization to remove first-order bias from nuisance estimation. We also derive an excess worst-case regret bound over a restricted policy class. The proposed method is validated through extensive simulations and an application to the 2022 Survey of Income and Program Participation (SIPP) data.
\end{abstract}


\clearpage
\section{Introduction}

Policy learning aims to assign treatments according to individuals' characteristics to maximize population welfare. In many applications, policy decisions depend on outcomes that are ordinal: their categories have a natural ordering, but the distances between categories are not intrinsically defined \citep{volfovsky2015causal}.
Many existing policy learning methods target mean outcomes or conditional average treatment effects, typically in settings with continuous or binary outcomes \citep{zhao2012estimating, kitagawa2018should, athey2021policy}.
Applying such methods to ordinal outcomes requires assigning numerical values to the ordered categories. A conditional average treatment effect computed from these values depends on the chosen coding, which imposes distances between categories that the ordering alone does not justify \citep{lu2018treatment, lu2020sharp}.
Moreover, when applied to ordinal responses, these approaches can mask substantial heterogeneity in how individuals benefit or suffer from potential changes in outcomes, thereby resulting in policies that could be detrimental to individuals adversely affected by treatment.

Accounting for this heterogeneity requires a welfare criterion that distinguishes how individuals are affected by treatment. Standard welfare criteria in policy learning depend on the outcome under the assigned treatment, without separately valuing treatment according to an individual's pair of potential outcomes \citep{kitagawa2018should, athey2021policy}. \citet{cui2025policy} address treatment effect heterogeneity by measuring welfare in terms of the distribution of individual treatment effects rather than their average. \citet{ben2024policy} develop policy learning with asymmetric counterfactual utilities for binary outcomes. Extending a fully specified utility function over all pairs of ordinal potential outcomes would require utilities for up to $J^2$ pairs when there are $J$ outcome levels, motivating a simpler grouping based on whether treatment improves the outcome. Such a grouping builds on the notion of strict benefit, for example, \citet{lu2018treatment} derive sharp bounds on the probabilities of benefit and strict benefit for ordinal outcomes. 
We use this notion to partition the population into a strictly benefited group, for whom treatment produces a better ordinal outcome than control, and a non-benefited group, for whom treatment does not produce a better outcome. We incorporate utilities that differ across these groups into the learning of individualized treatment policies, acknowledging that the welfare consequences for benefited and non-benefited individuals may differ in magnitude and social relevance.

Evaluating welfare under this utility specification requires the conditional probability of strict benefit. Because only one potential outcome is observed for each individual, this probability is generally only partially identified, even when treatment is unconfounded given the observed covariates \citep{lu2018treatment, greenland2020causal}. We therefore adopt a minimax regret criterion: each candidate policy is evaluated by its largest welfare loss relative to the corresponding oracle policy over all values of the probability of strict benefit compatible with the observed data \citep{manski2004statistical,kallus2021minimax, pu2021estimating}.

Learning the policy that minimizes this worst-case regret introduces an estimation challenge. The objective involves nested maximum and minimum operations, which can make the target functional nonpathwise differentiable and prevent direct application of standard semiparametric debiasing methods \citep{tsiatis2006semiparametric}. One approach addresses nonsmooth decision losses through outcome weighting frameworks with surrogate losses, which are often continuous convex upper bounds on the original discontinuous loss \citep{zhao2012estimating, pu2021estimating, cui2025policy}. Although such surrogates facilitate policy optimization, they do not directly resolve the nonregularity arising from the maximum and minimum operations used to construct the partial identification bounds. To address this type of nonregularity, \citet{levis2025covariate} leveraged smoothing techniques to estimate bounds on average treatment effects with a binary instrumental variable, while \citet{whitehouse2025inference} developed a general softmax smoothing and debiasing framework for inference on optimal policy values and other irregular, nondifferentiable functionals. These developments motivate the combination of functional smoothing and debiasing for minimax regret policy learning with ordinal outcomes and heterogeneous utilities.

Our contributions are threefold.
First, we formulate population welfare using heterogeneous utilities for strict benefit and non-benefit, and develop a minimax regret policy criterion that accounts for ambiguity in the partially identified probability of strict benefit.
Second, we construct smooth approximations for the nonsmooth functional and impose Neyman orthogonality on the smoothed worst-case regret to remove first-order bias in nuisance estimation \citep{chernozhukov2018double}.
Third, we derive an asymptotic upper bound on the excess worst-case regret of the learned policy relative to the minimax regret policy within a restricted policy class.

This article is organized as follows.
Section~\ref{sec:preliminary} describes the goal of policy learning with heterogeneous utilities in the setting with ordinal outcomes.
Section~\ref{sec:heterogeneity} discusses the partial identification problem and Section~\ref{sec:orthog-smooth} details the orthogonal smoothing strategy for learning the minimax regret policy with observed data.
Theoretical results for the asymptotic bound of excess worst-case regret are given in Section~\ref{sec:reg-converg}.
Sections~\ref{sec:sim} and~\ref{sec:app} evaluate the framework through simulations and an application to the 2022 Survey of Income and Program Participation (SIPP) data, respectively. Section~\ref{sec:discuss} concludes.


\section{Preliminaries}\label{sec:preliminary}

Suppose that we have access to observational data with $n$ samples of $ O = (X, A, Y) \in \mathcal{O} $, independently and identically drawn from a superpopulation characterized by the joint distribution $ \mathbb{P} $.
Here $X \in \mathcal{X} \subseteq \mathbb{R}^d$ is a vector of baseline covariates, $ A \in \{ 0, 1 \} $ is the treatment variable, where $ A = 1 $ corresponds to receiving the active treatment and $ A = 0 $ denotes control (or no treatment), and $ Y \in \mathcal{J} := \{ 0, 1, \ldots, J - 1 \}$ is an ordinal outcome with $J$ levels, where 0 and $J - 1$ represent the worst and best categories, respectively.
Our goal is to learn an optimal deterministic policy $ \pi: \mathcal{X} \rightarrow \{ 0, 1 \} $ to maximize population welfare with ordinal outcomes.
We adopt the potential outcome framework \citep{rubin1974estimating} and invoke the stable unit treatment value assumption \citep{rubin1980randomization}.
The following assumption is maintained throughout this article.

\begin{assumption}\label{assump:unconf-pos}
(i) (Unconfoundedness) $ Y_{a} \indep A \mid X $ for $ a = 0, 1 $; (ii) (Positivity) $ \prob ( A = a \mid X ) > 0 $ for $ a = 0, 1 $ almost surely.
\end{assumption}

Now we distinguish individuals who benefit from treatment from those who do not, i.e., strictly benefited individuals with $ Y_1 > Y_0 $, and non-benefited individuals with $ Y_1 \leq Y_0 $.
To this end, the population welfare (or, value function) that accounts for heterogeneous subgroups under a given policy $ \pi $ is expressed as
\begin{equation}\label{eq:asym-value}
\begin{aligned}
V(\pi) = \ & \EP \Big[ \Big\{ \pi(X) u_1^b (X) + \left( 1 - \pi(X) \right) u_0^b (X) \Big\} \cdot \indicator (Y_1 > Y_0) \\
& \quad\quad + \Big\{ \pi(X) u_1^n (X) + \left( 1 - \pi(X) \right) u_0^n (X) \Big\} \cdot \indicator (Y_1 \leq Y_0) \Big],
\end{aligned}
\end{equation}
where $u_a^b(X)$ and $u_a^n(X)$ are prespecified, bounded utility functions that depend on covariates $X$, received by each individual in the benefited and non-benefited groups under treatment $A = a$, respectively.
Here $\indicator\{\cdot\}$ denotes the indicator function.
We write $\EP$ for expectation under $\mathbb{P}$, and the superscripts ``b'' and ``n'' represent ``benefited'' and ``non-benefited,'' respectively.

In the binary case, the benefited group corresponds to \((Y_1,Y_0)=(1,0)\), and the remaining three strata are combined as the non-benefited group.
This formulation simplifies the otherwise complex stratification induced by ordinal outcome categories.
To focus on key ideas, we assume that utility functions in \eqref{eq:asym-value} are known throughout this paper.

\begin{remark}
When the benefited and non-benefited groups are assigned identical utilities under treatment $ A = a $, given by $ u_a^b (X) = u_a^n (X) = \E (Y_a \mid X)$, the population welfare \eqref{eq:asym-value} simplifies to $ V(\pi) = \EP \left[ (1 - \pi(X)) Y_0 + \pi (X) Y_1 \right], $ which leads to a specification that does not distinguish utilities across the benefited and non-benefited groups \citep{zhao2012estimating, pu2021estimating, zhang2025optimal}.
In contrast, our proposed value function \eqref{eq:asym-value} incorporates a richer structure that explicitly accounts for heterogeneity across distinct subgroups.
\end{remark}

Denote the probability of strict benefit conditional on $ X $ by $ \eta (X) = \prob ( Y_1 > Y_0 \mid X ) $, which is not point identified in general under Assumption~\ref{assump:unconf-pos} since it involves the joint distribution of $ (Y_1, Y_0) $ \citep{lu2018treatment, gabriel2024sharp}.
Define $ u^b (x) = u_1^{b} (x) - u_0^{b} (x) $ as the utility gained from receiving treatment relative to control in the benefited group for $ x \in \mathcal{X} $.
Similarly, $ u^n (x) = u_1^n (x) - u_0^n (x) $ for the non-benefited group represents the utility loss incurred from receiving treatment compared to control.
Conditioning on \(X=x\), the optimal treatment decision maximizes the conditional welfare, and the resulting oracle policy is
\begin{equation}\label{eq:oracle-pi}
\begin{aligned}
\pi^{\oracle} (x) = \indicator \Big\{ \left\{ u^b (x) - u^n (x) \right\} \eta (x) + u^n (x) > 0 \Big\}.
\end{aligned}
\end{equation}
If the contrast term between the utility gain and loss across two groups, $ u^b (x) - u^n (x) $, does not equal zero, the oracle policy $ \pi^{\oracle} $ need not be point identified in the absence of additional assumptions, as it depends on the unknown probability of strict benefit $ \eta (x) $.

Throughout this paper, we assume $ u^b(x) - u^n(x) > 0 $ for all $x \in \mathcal{X}$. The case in which $ u^b(x) - u^n(x) < 0 $ for all $x \in \mathcal{X}$ can be treated analogously, with the inequality in the oracle treatment rule reversed.
Therefore, the oracle policy $ \pi^{\oracle} (x)$ in \eqref{eq:oracle-pi} can be reformulated as
\begin{equation}\label{eq:asym-shift}
\pi^{\oracle} (x) = \indicator \{ \eta (x) > C_u (x) \}, ~\text{ where } C_u (x) := \dfrac{ -u^n (x) } { u^b (x) - u^n (x) }.
\end{equation}
Here, $C_u(x)$ is the utility threshold for the probability of strict benefit: treatment has positive expected incremental utility for individuals with covariates $X=x$ if and only if $\eta(x)>C_u(x)$.
Given \(X=x\), if \(C_u(x)<0\), then \(\pi^{\oracle}(x)=1\); if \(C_u(x)\geq1\), then \(\pi^{\oracle}(x)=0\). The decision can therefore depend on \(\eta(x)\) only when \(0\leq C_u(x)<1\), which is equivalent to \(u^b(x)>0\geq u^n(x)\). 
In this case, if $u^b(x)=1$ and $u^n(x)=-1$, then the utility loss and gain have equal magnitudes, yielding $C_u(x)=1/2$.
Thus, greater relative harm requires a higher probability of strict benefit to justify treatment.

\section{Policy Learning with Heterogeneous Subgroups}\label{sec:heterogeneity}

\subsection{Partial identification and worst-case regret}

Since the probability of strict benefit $\eta$ is generally not point identified under Assumption~\ref{assump:unconf-pos}, direct optimization of the value function based on the observed data is generally infeasible.
To quantify policy performance under the partial identification framework, we measure regret as $V(\pi^{\oracle}) - V(\pi)$, which is the difference between the population welfare under the oracle policy $\pi^{\oracle}$ and that under $ \pi $ \citep{manski2004statistical}. Specifically,
\begin{equation}\label{eq:regret}
\begin{aligned}
R (\pi; \eta) = \EP \left[ \left( \pi^{\oracle} (X) - \pi (X) \right) \left\{ \eta (X) \left( u^b (X) - u^n (X) \right) + u^n (X) \right\} \right].
\end{aligned}
\end{equation}
The expected utility loss \eqref{eq:regret} extends standard regret to ordinal outcomes by incorporating the unobserved probability of strict benefit $ \eta (X) $; see \cite{ben2024policy} for a related loss in the binary outcome case.

Let $ m_{0, k} (a, X) = \prob (Y_a = k \mid X) $ denote the true conditional probability of the potential outcome $ Y_a $ given $ X $ for $k \in \{ 0, 1, \ldots, J - 1 \}$.
The sharp bounds of $ \eta (X) $ are given by \cite{lu2018treatment}:
\begin{equation}\label{eq:eta-bounds}
\begin{gathered}
\eta_L (X) = \max_{0 \leq j \leq J - 1} \left\{ \sum_{k = j}^{J - 1} m_{0, k} (1, X) - \sum_{k = j}^{J - 1} m_{0, k} (0, X) \right\}, \\
\eta_U (X) = 1 + \min_{0 \leq j \leq J - 1} \left\{ \sum_{k = j + 1}^{J - 1} m_{0, k} (1, X) - \sum_{k = j}^{J - 1} m_{0, k} (0, X) \right\},
\end{gathered}
\end{equation}
where the empty sum equals zero.
Under Assumption~\ref{assump:unconf-pos}, $m_{0,k}(a,X)=\prob(Y=k\mid A=a,X)$ is point identified from the observed data, hence the bounds $\eta_L$ and $\eta_U$ are identifiable as well. 

Since the true $\eta(X)$ is unknown but lies within the interval $[\eta_L(X), \eta_U(X)]$, we evaluate each policy $\pi$ with its worst-case regret over all possible values of $\eta$.
Specifically, we define the worst-case regret as
\begin{equation*}
R_{\sup}(\pi) := \sup_{\eta: ~\eta_L \preceq \eta \preceq \eta_U} R(\pi; \eta),
\end{equation*}
where $\eta_L \preceq \eta \preceq \eta_U$ means $\eta_L(x) \leq \eta(x) \leq \eta_U(x)$ for all $x \in \mathcal{X}$.
We derive the optimal minimax policy that minimizes this worst-case regret in Section~\ref{sec:minimax-policy}.
The following proposition establishes that the worst-case regret $ R_{\sup} (\pi) $ can be identified from the observed data, despite the partial identification of the regret $ R (\pi; \eta) $ itself.

\begin{proposition}\label{prop:max-regret}
Under Assumption~\ref{assump:unconf-pos}, the worst-case regret of a policy $ \pi $ is
\begin{equation*}
\begin{aligned}
R_{\sup} (\pi) = C - \EP \Big[ \pi(X) \left\{ u^b (X) - u^n (X) \right\} \psi (X) \Big],
\end{aligned}
\end{equation*}
where $C$ is a constant that does not depend on $\pi$, and 
\begin{equation}\label{eq:score}
\begin{aligned}
\psi (X) := \psi_L (X) + \psi_U (X),
\end{aligned}
\end{equation}
where $ \psi_L (X) := \min \left\{ \eta_L (X) - C_u (X), 0 \right\} $ and $ \psi_U (X) := \max \left\{ \eta_U (X) - C_u (X), 0 \right\} $.
\end{proposition}

Proposition~\ref{prop:max-regret} suggests that applying heterogeneous utilities for different subgroups centers the bounds of $ \eta (X) $ at the threshold $ C_u (X) $ in $ \psi (X) $, which is a distinct feature of the expected worst-case utility loss $ R_{\sup} (\pi) $ in our framework, compared to existing literature in policy learning \citep{pu2021estimating, cui2025policy}.

\subsection{Minimax regret policy}\label{sec:minimax-policy}

Building on Proposition~\ref{prop:max-regret}, we now develop a minimax regret approach for policy learning under partial identification.
The goal is to find an optimal policy $\pi^*$ within a given policy class $\Pi$ that minimizes the worst-case expected utility loss relative to the oracle policy $\pi^{\oracle}$.
Formally, the minimax regret policy is defined as
\begin{equation}\label{eq:population-problem}
\pi^* \in \mathop{\arg \min}\limits_{\pi \in \Pi} ~R_{\sup} (\pi) = \mathop{\arg \min}\limits_{\pi \in \Pi} \sup_{\eta: ~\eta_L \preceq \eta \preceq \eta_U} R (\pi; \eta).
\end{equation}
We do not require $ \pi^* $ to be unique in $ \Pi $ since the performance of any policy $ \pi $ is evaluated via $ R_{\sup} (\pi) $ throughout our work.
Minimax regret formulations \eqref{eq:population-problem} have been studied for treatment choice under sampling uncertainty and partial identification \citep{manski2007minimax, stoye2009minimax, kallus2021minimax, ben2024policy, dadamo2021orthogonal}.
Our formulation considers ordinal outcomes and utilities that differ between benefited and non-benefited groups, and the probability of strict benefit is generally not point identified even in randomized trials.

When the policy class $\Pi$ is unrestricted, Proposition~\ref{prop:max-regret} yields the minimax policy $ \pi^* $ as the solution to problem \eqref{eq:population-problem}:
\begin{equation}\label{eq:minimax-policy}
\pi^* (X) = \indicator \{ C_u (X) < 0 \} + \indicator \{ \psi (X) > 0 \} \cdot \indicator \{ 0 \leq C_u (X) < 1 \}.
\end{equation}
where $ \psi (X) $ is defined in \eqref{eq:score}.
When $u^b(X)>0\geq u^n(X)$, $C_u(X)\in[0,1)$, then \eqref{eq:minimax-policy} reduces to $\pi^*(X)=\indicator\{\psi(X)>0\} = \indicator\{(\eta_L(X)+\eta_U(X))/2>C_u(X)\}$, indicating that the unrestricted minimax policy depends on the midpoint of the partial identification bounds of $\eta$. 
Intuitively, the minimax regret policy prioritizes correct treatment allocation to individuals who explicitly gain benefit from (or suffer loss by) the treatment.
See Figure~\ref{fig:OTR} for a schematic illustration.

\begin{figure}[htbp]
\vspace{0.8\baselineskip}
\centering
\resizebox{\linewidth}{!}{\begin{tikzpicture}[>=stealth]

\begin{scope}[xshift=0cm]
\draw[->] (-2,0) -- (2,0);
\draw (-1.0,0.1) -- (-1.0,-0.1);
\draw (1.4,0.1) -- (1.4,-0.1);
\node[below] at (-1.0,-0.1) {$0=\psi_L(X)$};
\node[below] at (1.4,-0.1) {$\psi_U(X)$};
\node[below=1cm] at (0,0) {(a) $\pi^*(X)=1$};
\end{scope}

\begin{scope}[xshift=6cm]
\draw[->] (-2,0) -- (2,0);
\draw (-1,0.1) -- (-1,-0.1);
\draw (0,0.1) -- (0,-0.1);
\draw (1.2,0.1) -- (1.2,-0.1);
\node[below] at (-1,-0.1) {$\psi_L(X)$};
\node[below] at (0,-0.1) {$0$};
\node[below] at (1,-0.1) {$\psi_U(X)$};
\node[below=1cm] at (0,0)
  {(b) $\pi^*(X)=\indicator\{\psi(X)>0\}$};
\end{scope}

\begin{scope}[xshift=12cm]
\draw[->] (-2,0) -- (2,0);
\draw (-1.4,0.1) -- (-1.4,-0.1);
\draw (1.0,0.1) -- (1.0,-0.1);
\node[below] at (-1.4,-0.1) {$\psi_L(X)$};
\node[below] at (1.0,-0.1) {$\psi_U(X)=0$};
\node[below=1cm] at (0,0) {(c) $\pi^*(X)=0$};
\end{scope}

\end{tikzpicture}}
\caption{Illustration of the optimal policy under partial identification.}
\label{fig:OTR}
\end{figure}
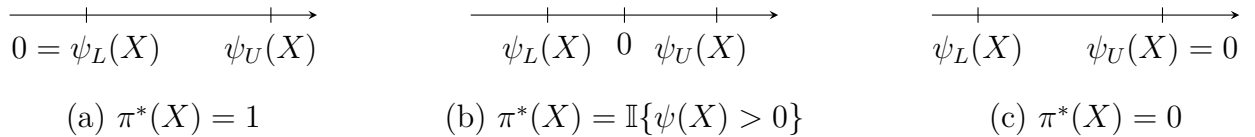

\section{Learning a Policy via Orthogonal Smoothing}\label{sec:orthog-smooth}

{Learning the minimax regret policy in \eqref{eq:population-problem} requires estimating the score $\psi(X)$, whose nested maximum and minimum operations can induce plug-in bias and prevent direct application of standard semiparametric debiasing methods \citep{chernozhukov2013intersection, chernozhukov2018double}. In this section, we address these difficulties by smoothing the score and adding a Neyman orthogonal correction to remove the first-order effect of nuisance estimation \citep{levis2025covariate, whitehouse2025inference}.}

For \(a\in\{0,1\}\) and \(j\in\mathcal J\), let \(m_j(a,x)\) and \(e(x)\) denote working models for the true conditional outcome probability \(m_{0,j}(a,x)\) and propensity score \(e_0(x) = \mathbb{P} (A = 1 \mid X=x)\), respectively. Write \(m(a,x)=\{m_j(a,x):j\in\mathcal J\}\), and define \(e(a,x)=ae(x)+(1-a)\{1-e(x)\}\). Let \(\theta=(m,e)\), with true value \(\theta_0=(m_0,e_0)\), and let \(\widehat\theta=(\widehat m,\widehat e)\) denote its estimator.

\subsection{Smooth approximation to the worst-case regret}\label{sec:smooth-approx}

For a vector $g(x)=(g_1(x),\ldots,g_K(x))$, we approximate $G(x)=\max_{1\leq k\leq K}g_k(x)$ by
\begin{equation*}
G^{\beta_n} (g(x)):= \frac{ \sum_{k = 1}^K g_k (x) \exp(\beta_n g_k (x)) }{ \sum_{i = 1}^K \exp(\beta_n g_i (x)) },
\end{equation*}
where $\beta_n>0$ is a smoothing parameter that may depend on the sample size $n$. As $\beta_n\to\infty$, we have $G^{\beta_n}(g(x))\to G(x)$.
Properties of $ G^{\beta_n} $ are shown in Section~\ref{app:sm-property} of the Supplementary Material.

Let \(\eta_L(x;m)\) and \(\eta_U(x;m)\) denote the bound functionals in \eqref{eq:eta-bounds} with \(m_{0,k}\) replaced by \(m_k\). For \(0\leq j\leq J-1\), define
\begin{equation*}
\begin{aligned}
\delta_{L,j}(x;m)&:=\sum_{k=j}^{J-1}m_k(1,x)-\sum_{k=j}^{J-1}m_k(0,x)-C_u(x),\\
\delta_{U,j}(x;m)&:=1+\sum_{k=j+1}^{J-1}m_k(1,x)-\sum_{k=j}^{J-1}m_k(0,x)-C_u(x).
\end{aligned}
\end{equation*}
Let $\delta_L(x;m)$ and $\delta_U(x;m)$ denote the corresponding $J$-dimensional vectors. Then $\eta_L(x;m)-C_u(x)=\max_{j\in\mathcal J}\delta_{L,j}(x;m)$ and $\eta_U(x;m)-C_u(x)=\min_{j\in\mathcal J}\delta_{U,j}(x;m)$.
The smoothed upper and lower contributions in $ \psi (x) $ defined in \eqref{eq:score} are 
\begin{equation}\label{eq:softmax-approx}
	\begin{aligned}
		\psi_L^{\beta_n} (x;m) &= G^{\beta_n}(\delta_L(x;m)) - G^{\beta_n}(\{\delta_L(x;m),0\}),\\
		\psi_U^{\beta_n} (x;m) &= G^{\beta_n}(\{-\delta_U(x;m),0\}) - G^{\beta_n}(-\delta_U(x;m)),
	\end{aligned}
\end{equation}
where the superscript $ \beta_n $ denotes the smoothed versions of corresponding functions, and $(\nu,0)$ denotes the vector $\nu$ appended with a zero.
Then the smoothed version of $ \psi (x) $ is $\psi^{\beta_n} (x;m) = \psi_L^{\beta_n} (x;m) + \psi_U^{\beta_n} (x;m)$.

Because the constant \(C\) in Proposition~\ref{prop:max-regret} is independent of \(\pi\), hereafter we omit this common additive term from both \(R_{\sup}\) and its smoothed counterpart, without affecting policy optimization. Accordingly, define
$$ R_{\sup}^{\beta_n} (\pi; \theta_0) := - \EP \left[ \left\{ u^b (X) - u^n (X) \right\} \pi(X) \psi^{\beta_n} (X) \right].$$ 
Hereafter, we explicitly indicate dependence on nuisance functions by writing key quantities as functions of $ \theta $; e.g., $ R_{\sup}^{\beta_n} (\pi) = R_{\sup}^{\beta_n} (\pi; \theta_0) $.
The following proposition shows that for an arbitrary policy $\pi$, the approximation error can be controlled by the growth rate of smoothing parameters $\{ \beta_n \}_{n \geq 1}$.

\begin{proposition}\label{prop:approx-error}
Let $\{ \beta_n \}_{n \geq 1}$ be a sequence of positive smoothing parameters.
Given a policy $ \pi \in \Pi $, the approximation error between the smoothed and original worst-case regret satisfies
\begin{equation*}
R_{\sup}^{\beta_n} (\pi; \theta_0)- R_{\sup} (\pi; \theta_0) = O (\beta_n^{-1}).
\end{equation*}
\end{proposition}

{Proposition~\ref{prop:approx-error} shows that the smoothing approximation error vanishes as $\beta_n\to\infty$. However, the smoothed criterion still depends on estimated nuisance functions, and increasing $\beta_n$ can magnify the second-order outcome model remainder. We next construct an orthogonal score to remove first-order nuisance bias; Section~\ref{sec:reg-converg} then gives rate conditions balancing the remaining nuisance error against smoothing approximation error.}

\subsection{Neyman orthogonality for robust estimation}\label{sec:orthog}

The smoothed worst-case regret $ R_{\sup}^{\beta_n} (\pi; \theta_0) $ constructed in Section \ref{sec:smooth-approx} depends on the true conditional outcome probabilities $ m_{0} $, which are generally unknown and must be estimated from the observed data.
A Neyman orthogonal score recovers the target in expectation at $\theta_0$ and has zero first derivative with respect to the nuisance functions \citep{chernozhukov2018double}. Under suitable smoothness conditions, its nuisance remainder is therefore second-order, with a magnitude that can depend on $\beta_n$, as shown later in Section~\ref{sec:reg-converg}. The expansion and formal definition are given in Section~\ref{app:orthog-background} of the Supplementary Material.

{To construct the Neyman orthogonal score, define the influence function correction term \citep{chernozhukov2018double}}
\begin{equation}\label{eq:dr-correction}
T_{a, j} (O; \theta) = \frac{ \indicator \{ A = a \} }{ e (a, X) } \left\{ \indicator \{ Y = j \} - m_j (a, X) \right\}.
\end{equation}
{Aggregating these residuals according to the lower and upper bound formulas in \eqref{eq:eta-bounds} gives}
\begin{equation*}
\begin{aligned}
L_j (O; \theta) &= \sum_{k = j}^{J - 1} T_{1, k} (O; \theta) - \sum_{k = j}^{J - 1} T_{0, k} (O; \theta), \\ U_j (O; \theta) &= \sum_{k = j + 1}^{J - 1} T_{1, k} (O; \theta) - \sum_{k = j}^{J - 1} T_{0, k} (O; \theta).
\end{aligned}
\end{equation*}

\begin{proposition}\label{prop:orthogonal-score}
For any $\beta_n > 0$, the Neyman orthogonal score of $ R_{\sup}^{\beta_n} (\pi; \theta) $ is given by
\begin{equation}\label{eq:NO-score}
\Psi^{\beta_n} (O, \pi; \theta) = -\{u^b(X)-u^n(X)\} \pi (X) \left\{ \psi^{\beta_n} (X; m) + \phi^{\beta_n} (O; \theta) \right\},
\end{equation}
where $ \psi^{\beta_n} (X; m) $ is defined in \eqref{eq:softmax-approx}, and $ \phi^{\beta_n} (O; \theta) $ is the first-order bias correction term:
\begin{equation*}
\phi^{\beta_n} (O; \theta) = \sum_{j = 0}^{J - 1} \frac{ \partial \psi_L^{\beta_n} (X; m) }{ \partial \delta_{L, j} } L_j (O; \theta) + \sum_{j = 0}^{J - 1} \frac{ \partial \psi_U^{\beta_n} (X; m) }{ \partial \delta_{U, j} } U_j (O; \theta).
\end{equation*}
\end{proposition}

Proposition~\ref{prop:orthogonal-score} shows that the smoothed score $\psi^{\beta_n}$ can be augmented by the correction term $\phi^{\beta_n}$ to remove its first-order sensitivity to nuisance estimation. 
Write $R_{\sup}^{\beta_n,\no}(\pi;\theta):=\EP[\Psi^{\beta_n}(O,\pi;\theta)]$ for the population criterion induced by this score. At $\theta=\theta_0$, $R_{\sup}^{\beta_n,\no}(\pi;\theta_0)=R_{\sup}^{\beta_n}(\pi;\theta_0)$.
By invoking Proposition~\ref{prop:approx-error}, we conclude that given policy $ \pi \in \Pi $, $R_{\sup}^{\beta_n, \no}(\pi; \theta_0)$ converges to the true worst-case regret $R_{\sup}(\pi; \theta_0)$ as $\beta_n \to \infty$.

\subsection{Estimation of the minimax regret policy}

{Let $\Pi_n$ be a restricted policy class that may depend on $n$, and let $\pi_n^*\in\arg\min_{\pi\in\Pi_n}R_{\sup}(\pi)$ be its minimax regret policy. Given $\beta_n>0$, we obtain the orthogonal smoothed estimator $\hat\pi_n$ as follows.}
{The procedure can also be implemented using cross-fitting \citep{chernozhukov2018double}.}

{
\begin{enumerate}[label=\textit{Step \arabic*.}, leftmargin=*]
\item
Estimate the propensity score and conditional outcome probabilities on a sample independent of the observations $\{O_i\}_{i=1}^n$, obtaining $\hat\theta=(\hat m,\hat e)$.

\item
Evaluate the score in \eqref{eq:NO-score} using $\hat\theta$ and form the empirical criterion
\begin{equation*}
\hat R_{\sup}^{\beta_n,\no}(\pi;\hat\theta)
=\Pn\Psi^{\beta_n}(O,\pi;\hat\theta)
=\frac{1}{n}\sum_{i=1}^n\Psi^{\beta_n}(O_i,\pi;\hat\theta),
\end{equation*}
where $\Pn f=n^{-1}\sum_{i=1}^n f(O_i)$ denotes the empirical average.

\item
Minimize this criterion over $\Pi_n$ and obtain
\begin{equation*}
\hat\pi_n\in\arg\min_{\pi\in\Pi_n}\hat R_{\sup}^{\beta_n,\no}(\pi;\hat\theta).
\end{equation*}
\end{enumerate}
}

\section{Statistical guarantees for the estimated policy}\label{sec:reg-converg}

This section establishes asymptotic convergence guarantees for $ \hat{\pi}_n $ to $\pi^*_n$ \eqref{eq:population-problem} in terms of excess risk by analyzing second-order bias and identifying growth rates of $\beta_n$ sufficient for root-$n$ consistency of regret estimation.
To control the statistical error from optimizing over $\Pi_n$, we impose the following complexity restriction.
\begin{assumption}\label{assump:class-complexity}
There exist constants $ \zeta_\Pi \in (0, 1/2) $ and $ N \geq 1 $ such that the Vapnik-Chervonenkis dimension of $ \Pi_n $ is bounded by $ \vc (\Pi_n) \leq n^{\zeta_\Pi} $ for $ n \geq N $.
\end{assumption}
Assumption~\ref{assump:class-complexity} allows the policy class complexity to grow moderately with the sample size $n$.
Common policy classes, including linear rules and decision trees can satisfy this condition when their dimensions, depths, or architectures are appropriately controlled \citep{athey2021policy,zhou2023offline}.

For nonsmooth functionals, regular estimators do not exist without further assumptions \citep{hirano2012impossibility}.
In particular, when the maximizer or minimizer is nonunique with positive probability, the resulting functional may lack pathwise differentiability, complicating standard regular asymptotically linear estimation.
To address this challenge, we impose the following margin condition, which controls the probability mass near the optimality boundary \citep{chen2023inference, luedtke2016statistical}.

\begin{assumption}\label{assump:margin-cond}
There exist constants $\delta>0$ and $C_0>0$ such that
\[
	\mathbb P\left(v_{(1)}(X)-v_{(2)}(X)\leq t\right)\leq C_0t^\delta,
\]
for every vector $v\in\{\delta_L,(\delta_L,0),-\delta_U,(-\delta_U,0)\}$ and every $t\geq0$, where \(v_{(1)}\geq v_{(2)}\) are the largest and second-largest components of $v$.
\end{assumption}

Assumption~\ref{assump:margin-cond} imposes a margin condition on the probability of a near-tie to decay at least as fast as $t^\delta$ as $t$ approaches $0$, with larger $\delta$ imposing stronger separation near the optimality boundary.
In fact, setting $t = 0$ for any $\delta > 0$ implies that the maximizer of each vector is almost surely unique, which avoids issues that arise when ties occur with positive probability \citep{luedtke2016statistical}.
The following theorem shows how policy class complexity, smoothing approximation, and nuisance estimation jointly determine the excess worst-case regret of the learned policy. 

\begin{theorem}\label{cor:regret-converge}
{
Suppose Assumptions~\ref{assump:unconf-pos}--\ref{assump:margin-cond} hold and $\beta_n\to\infty$. Let $\hat e$ and $\hat m_j$, $j\in\mathcal J$, be estimated on an independent sample and satisfy
\[
\|\hat e-e_0\|=o_{\mathbb P}(1),\qquad
\|\hat m_j-m_{0,j}\|=o_{\mathbb P}(1),\quad j\in\mathcal J.
\]
Suppose $\|\hat m(A,X)\|_{L_\infty(\mathbb P)}<B$ for some constant $B>0$, and $\mathbb P\{\epsilon\leq\hat e(X)\leq1-\epsilon\}=1$ for some $\epsilon>0$. Then the excess worst-case regret of $\hat\pi_n$ relative to $\pi_n^*$ satisfies}
\begin{equation}\label{eq:cor1}
\begin{aligned}
& R_{\sup} (\hat{\pi}_n; \theta_0) - R_{\sup} (\pi_n^*; \theta_0)\\ =\ & {O_{\mathbb{P}} \left( \sqrt{\vc(\Pi_n) / n} \right)}+ {O_{\mathbb{P}} \left( || \hat{e} - e_0 || \cdot \max_{0 \leq j \leq J - 1} || \hat{m}_{j} - m_{0, j} || \right)} \\
&+ {O_{\mathbb{P}} \left( \beta_n^{- (1 + \delta)} \right)} + {O_{\mathbb{P}} \left( \beta_n \max_{0 \leq j \leq J - 1} || \hat{m}_{j} - m_{0, j} ||^2 \right)} + o_{\mathbb{P}} (n^{-1 / 2}).
\end{aligned}
\end{equation}
\end{theorem}

 Theorem~\ref{cor:regret-converge} gives the bound for the excess regret, which contains four components apart from a negligible remainder. The first term captures sampling variability from optimizing over $\Pi_n$ and reflects the uniform control over the policy class needed to extend the fixed policy analysis to a learned policy, as established in Theorem~\ref{thm:asymp-converg} of the Supplementary Material. The second term is the product of propensity score and outcome model estimation errors. The third is smoothing approximation error, and the fourth is the squared outcome model error multiplied by $\beta_n$. Thus, orthogonalization removes first-order nuisance bias, while smoothing leaves an additional second-order term beyond the usual product of nuisance errors.

{These terms reveal the trade-off in choosing $\beta_n$: increasing it reduces approximation error but magnifies the outcome model remainder. A simpler policy class reduces the sampling term, but does not remove the smoothing and nuisance errors. This separates the role of policy class complexity from the accuracy needed to estimate the policy criterion.}

Next, we give sufficient growth rates of $\beta_n$ for root-$n$ consistency according to Theorem~\ref{cor:regret-converge}. Let $\beta_n\asymp n^{\zeta_\beta}$, $\|\hat m_j-m_{0,j}\|=O_{\mathbb P}(n^{-\zeta_m})$, and $\|\hat e-e_0\|=O_{\mathbb P}(n^{-\zeta_e})$. Here, $a_n\asymp b_n$ means $c b_n\leq a_n\leq C b_n$ for constants $0<c\leq C<\infty$ and all sufficiently large $n$.
To ensure the excess-regret bound is $O_{\mathbb P}(\sqrt{\vc(\Pi_n)/n}+n^{-1/2})$, sufficient rate conditions are
\begin{equation}\label{eq:rate-constraint}
\frac{1}{2(1+\delta)}\leq\zeta_\beta\leq2\zeta_m-\frac12,
\qquad \zeta_e+\zeta_m\geq\frac12,
\end{equation}
where $\delta>0$ is the exponent in Assumption~\ref{assump:margin-cond}. The lower and upper bounds on $\zeta_\beta$ control the third and fourth terms in \eqref{eq:cor1}, respectively, while the final inequality controls the second term. A feasible choice of $\zeta_\beta$ requires $\zeta_m\geq(2+\delta)/\{4(1+\delta)\}>1/4$. By improving the smoothing error from $O(\beta_n^{-1})$ to $O(\beta_n^{-(1+\delta)})$, the margin condition relaxes the outcome model rate requirement relative to the $\zeta_m\geq1/2$ restriction obtained from the general approximation bound, provided the propensity score estimator converges sufficiently fast.



\section{Simulation}\label{sec:sim}
\subsection{Data generating process}

In this section, we evaluate the finite sample performance of the proposed method from two aspects: the worst-case regret $R_{\sup}$ and the sensitivity analysis regarding the utility threshold $C_u$.
We compare the direct plug-in estimator and the proposed orthogonal smoothed estimator with two additional estimators that separate the roles of orthogonalization and smoothing.
The direct influence function estimator augments the unsmoothed plug-in score using the influence function corrections corresponding to the optimal branches of the maximum and minimum operations \citep{levis2025covariate}.
The smoothed plug-in estimator uses the same smooth approximation as the orthogonal smoothed estimator but omits the orthogonal correction term.
Table~\ref{tab:estimators} summarizes the construction of the four estimators in this section.

\begin{table}[htbp]
\renewcommand{\baselinestretch}{1}\small
\centering
\caption{Estimators compared in the simulation of $R_{\sup}$.
IF denotes influence function, and checkmarks indicate the components included in each estimator.}
\label{tab:estimators}
\begin{tabular}{@{}lcc@{}}
\toprule
Estimator & Smoothing & IF correction \\
\midrule
Direct plug-in & $\times$ & $\times$ \\
Direct IF & $\times$ & $\checkmark$ \\
Smoothed plug-in & $\checkmark$ & $\times$ \\
Orthogonal smoothed & $\checkmark$ & $\checkmark$ \\
\bottomrule
\end{tabular}
\end{table}

Our data generating process follows a similar setup to \cite{levis2025covariate}.
Two baseline covariates $ X = (X_1, X_2) $ are generated independently from $ \operatorname{Unif}(-1,1) $, and $ A \mid X \sim \operatorname{Bernoulli}\{e(X)\} $, where $e(X)=\min\{\max(X_1^2,0.1),0.9\}$.
The ordinal outcome is generated from the multinomial logistic model provided in Section~\ref{app:add-sim} of the Supplementary Material, with $ J \in \{3,5,8\} $ outcome levels.
For the primary worst-case regret analysis, we set the utility threshold to a constant $ C_u = 0.35 $, and subsequently vary $C_u$ in the sensitivity analysis.
Equivalently, we normalize \(u^b-u^n=1\), so that \(u^b=1-C_u\) and \(u^n=-C_u\).

Following \cite{kennedy2023towards} and \cite{levis2025covariate}, we construct controlled nuisance estimators by perturbing the true nuisance functions on the logit scale.
For observation $i$ and $j=0,\ldots,J-1$, we set $\hat m_j(a,X_i)=\expit(\logit\{m_j(a,X_i)\}+\epsilon_{m,j,a,i}(n))$ and $\hat e(X_i)=\expit(\logit\{e(X_i)\}+\epsilon_{e,i}(n))$.
The perturbations are mutually independent, with $\epsilon_{m,j,a,i}(n)\sim N((2a-1)hn^{-r},h^2n^{-2r})$ and $\epsilon_{e,i}(n)\sim N(hn^{-r},h^2n^{-2r})$.
We set $ h = 2 $ and vary the error rate over $ r \in \{0.10+0.05k:k=0,\ldots,8\}$.
This construction yields $ \|\hat{e}-e\|=O_{\mathbb{P}}(n^{-r}) $ and $ \|\hat{m}_j-m_j\|=O_{\mathbb{P}}(n^{-r}) $, so that $ \zeta_e=\zeta_m=r $.
As $r$ increases, the convergence rate of nuisance components approaches the parametric rate.

We consider sample sizes $ n \in \{500,1000,5000\} $, and take $\beta_n=2h n^{\zeta_\beta}$, where its growing rate $ \zeta_\beta = \max \{ c_\beta, 2\zeta_m - 0.5 \} $, and we set $c_\beta=0.25$ to ensure that $\beta_n$ diverges throughout the simulated nuisance-rate range, so that the smoothing approximation error vanishes.
Each scenario is replicated 500 times, and independent training and test samples of size $n$ are used for policy learning and evaluation, respectively, in each replication.
We focus on the policy class of depth-2 decision trees
\citep{athey2021policy,ben2024policy,jin2025policy}.

Figure~\ref{fig:rsup-j3} displays average excess regret, absolute bias and root mean squared error (RMSE) for $J=3$ as the error rate $r$ varies. Results for $J=5$ and $J=8$ and corresponding numerical tables are provided in Section~\ref{app:add-sim} of the {Supplementary Material}.
In the simulations, we estimate $R_{\sup}(\widehat\pi_n;\theta_0)-R_{\sup}(\pi_n^*;\theta_0)$ on independent test samples and report its average across replications, where the population minimizer $\pi_n^*$ is approximated by the empirical oracle depth-2 tree fitted using the true training-sample scores. Absolute bias is the absolute value of the mean estimation error $\widehat R_{\sup}(\widehat\pi_n;\widehat\theta)-R_{\sup}(\widehat\pi_n;\theta_0)$ across Monte Carlo replications.
For each learned policy, bias and RMSE are calculated relative to the criterion based on true score evaluated on the same independent test sample, \(-\Pn (\hat\pi_n(X)\psi(X;m_0))\), which provides a Monte Carlo approximation to its population worst-case regret $R_{\sup}(\widehat\pi_n;\theta_0)$.
Numerical results are reported in Table~\ref{tab:rsup-j3} in the {Supplementary Material}.

\subsection{Policy performance and regret estimation}

According to Figure~\ref{fig:rsup-j3}, the direct plug-in and smoothed plug-in estimators exhibit nearly identical excess regret, absolute bias, and RMSE, while the direct influence function and orthogonal smoothed estimators also perform similarly. 
This indicates that the observed performance gains of corrected estimators (i.e., the direct influence function and orthogonal smoothed estimators) mainly come from the orthogonal correction terms, whereas smoothing provides little improvement in either regret estimation accuracy or policy performance under this setting.
Across much of the error rate range, the two corrected estimators achieve substantially lower errors than their uncorrected counterparts (i.e., the plug-in and smoothed plug-in estimators). Absolute bias closely tracks RMSE for all the four methods when nuisance convergence is slow, indicating that systematic bias accounts for most of the regret estimation error in these settings.

\begin{figure}[t]
\centering
\includegraphics[width=0.92\textwidth]{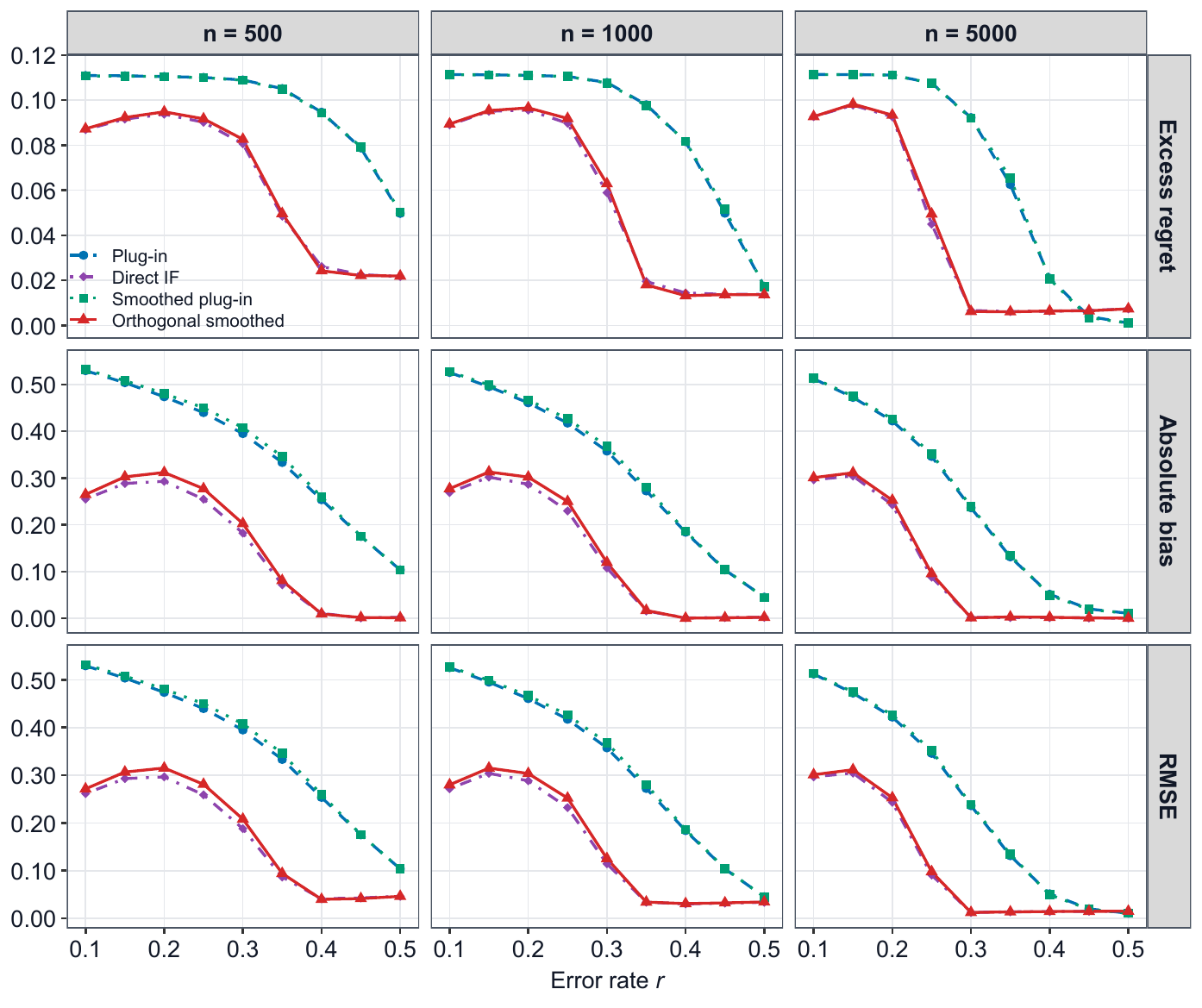}
\caption{Excess worst-case regret, absolute bias, and RMSE of regret estimation under different sample sizes and error rates. Here $J=3$.}
\label{fig:rsup-j3}
\end{figure}

As nuisance convergence becomes faster, that is, when $r$ increases, the uncorrected plug-in estimators generally improve more gradually, whereas both corrected estimators exhibit a pronounced decline in excess regret, absolute bias, and RMSE at intermediate nuisance rates.
This decline occurs at a lower nuisance rate as the sample size increases, appearing around \(r=0.3\) when $n=5000$. 
For the orthogonal smoothed estimator, the reduction in excess regret is consistent with the shrinking nuisance error contributions in the second and fourth terms of \eqref{eq:cor1} in Theorem~\ref{cor:regret-converge}.

At sufficiently fast nuisance convergence rates (e.g., $r>0.3$), the performance gap narrows as the sample size grows. When nuisance estimators converge at the parametric rate, both the direct and smoothed plug-in estimators become competitive and attain lower excess regret and RMSE when $n=5000$. Although the corrected estimators retain smaller absolute bias, their greater variability can outweigh this advantage once plug-in bias is sufficiently small. These results highlight the value of the influence function correction when nuisance error remains appreciable, while showing that more accurate nuisance estimation and larger samples can make uncorrected plug-in estimators competitive.

\subsection{Sensitivity to the utility threshold}
\label{sec:utility-sensitivity}

Figure~\ref{fig:trpr-n1000} presents a sensitivity analysis of treatment allocation to the utility threshold $C_u$, reporting the mean treated proportion as \(C_u\) varies, with sample size \(n=1000\) and nuisance convergence rate \(r=0.4\).
As in Figure~\ref{fig:rsup-j3}, the direct and smoothed plug-in estimators closely overlap, and the direct influence function and orthogonal smoothed estimators exhibit the same pattern.
All four methods approach the always-treat policy when $C_u$ is close to zero and always-control when $C_u$ is close to one, reflecting the increasing penalty for treating individuals who do not benefit relative to the gain from treating those who do.
The main transition shifts toward larger values of $C_u$ as $J$ increases.
This pattern is consistent with the unrestricted minimax rule, which assigns treatment when $C_u$ lies below the midpoint of the bounds on the probability of strict benefit, since this midpoint increases on average over the covariate distribution as \(J\) increases under the present data generating process, primarily through an increase in the upper bound.
Treatment therefore remains preferable under the minimax criterion at larger relative utility losses.
The corrected methods generally begin reducing treatment at lower thresholds and exhibit a more gradual transition than the plug-in methods.
In practice, researchers should consider a plausible range for the relative utility loss and gain and examine treatment recommendations throughout that range, paying particular attention to transition regions where modest changes in preferences can substantially alter treatment allocation.
Corresponding results for $n=500$ and $n=5000$ are provided in Section~\ref{app:add-sim} of the Supplementary Material.

\begin{figure}[htbp]
\centering
\includegraphics[width=0.96\textwidth]{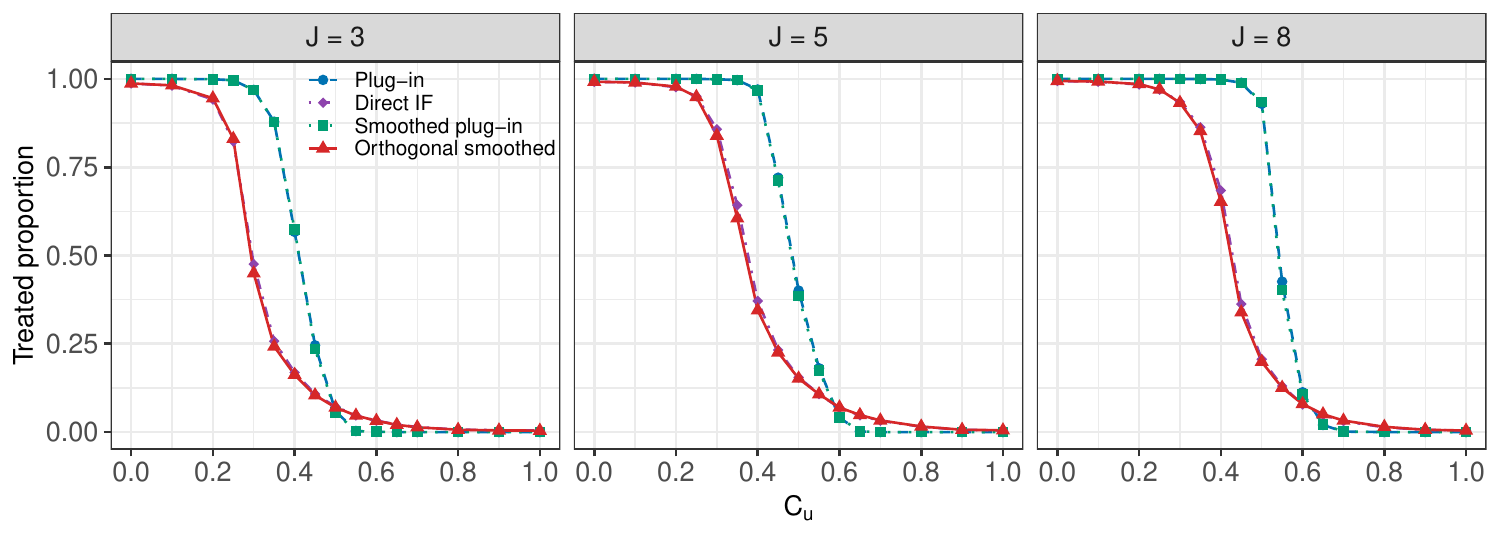}
\caption{Mean treated proportion as the utility threshold $C_u$ varies with sample size $n=1000$ and outcome levels $J=3$, 5, and 8.}
\label{fig:trpr-n1000}
\end{figure}

\section{Application to SIPP data}\label{sec:app}

In this section, we turn to a real-world dataset to illustrate the proposed methodology: the \href{https://www.census.gov/programs-surveys/sipp/data/datasets/2022-data/2022.html}{2022 Survey of Income and Program Participation (SIPP) data}, a longitudinal survey conducted by the U.S.\ Census Bureau. 
Of particular interest is a decision problem: among employed retirement-account owners, does holding an account through the main employer or business improve self-reported health?
Motivated by this question, we illustrate how the proposed method learns hypothetical allocation rules under different preferences over health benefits and non-benefits when health outcomes are ordinal.

In the SIPP data, the ordinal outcome $Y$ is an individual's self-reported health status, with 5 categories, from ``Poor'' with $ Y_i = 0 $ to ``Excellent'' with $ Y_i = 4 $.
The binary treatment $A_i = 1$ indicates that an individual held a 401 (k), 403 (b), 503 (b), or Thrift Savings Plan account provided through their main employer {or business} during the reference period, whereas $ A_i = 0 $ {indicates ownership of such an account without one provided through the main employer or business}.
In addition, we adjust for a collection of {observed} covariates $X$, including age, gender, race, {ethnicity,} education level, household income-to-poverty ratio, {marital status}, health insurance coverage, and work limitation, so that the covariates $X$ include rich demographic, socioeconomic, and health-related characteristics that {may be associated with} both participation in employer-sponsored retirement plans and self-reported health status.
{We retain all $n=8319$ complete case records for December 2021 from respondents who owned one of these accounts during 2021 and had at least one job that December. Each record corresponds to a distinct individual. The analytic sample contains 6624 treated observations ($A=1$) and 1695 control observations ($A=0$).}

{For all four estimators in Table~\ref{tab:estimators}, we estimate the propensity score using logistic regression with an intercept, linear and quadratic age terms, and main effects for gender, race, ethnicity, education, income-to-poverty ratio, marital status, health insurance coverage, and work limitation. Conditional outcome probabilities are estimated using multinomial logistic regression with the same terms and a treatment indicator. We learn depth-2 decision trees using threefold rotating sample splits, with equal sample sizes for nuisance estimation, policy learning, and evaluation, and summarize results over 100 bootstrap replications. Section~\ref{app:sipp-results} of the Supplementary Material gives the detailed procedure regarding this application.}

The proportions of individuals assigned the active treatment by each estimator under different $C_u$'s are reported in Figure~\ref{fig:sipp-trP-new}. 
We can see that the treated proportion varies greatly with the utility threshold $C_u$ in \eqref{eq:asym-shift}, which reflects the asymmetry of heterogeneous utilities, as shown in the left panel of Figure~\ref{fig:sipp-trP-new}. 
All four methods assign treatment to nearly everyone at the lower end of the threshold range and to nearly no one at the upper end, with the main transition occurring between approximately $0.25$ and $0.35$, as indicated by the right panel of Figure~\ref{fig:sipp-trP-new}. 
As $C_u$ decreases, the utility loss for the non-benefited group becomes less severe relative to the utility gain for the benefited group, and the treated proportion increases.
The direct and smoothed plug-in estimators display a sharper transition, whereas the direct influence function and orthogonal smoothed estimators begin reducing treatment at lower thresholds and retain larger treated proportions {toward the upper end of the transition region}.
The corrected estimators therefore exhibit a more gradual response to the utility threshold, consistent with the qualitative pattern in Section~\ref{sec:utility-sensitivity}.
Table~\ref{tab:sipp-Rsup-new} in Section~\ref{app:sipp-results} of the Supplementary Material provides the corresponding numerical summaries and bootstrap standard deviations, showing appreciable variability in treatment allocation near the transition region.

\begin{figure}[htbp]
\centering
\includegraphics[width=0.92\linewidth]{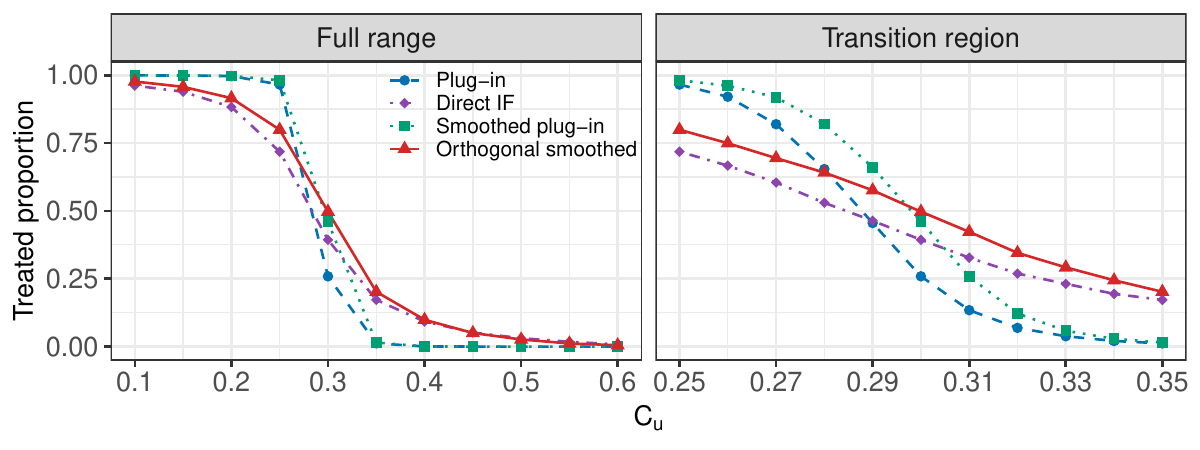}
\caption{{Treated proportions under different $C_u$'s for the December SIPP.}
The left panel displays $C_u\in[0.1,0.6]$, and the right panel details $C_u\in[0.25,0.35]$.}
\label{fig:sipp-trP-new}
\end{figure}

These results show how the proposed framework can inform recommendations for retirement-account ownership through the main employer or business among individuals who already own retirement accounts.
The orthogonal smoothed policy's mean treated proportion falls from approximately {$0.8$} at $C_u=0.25$ to {$0.5$} at $C_u=0.3$.
Thus, placing moderately greater weight on the utility loss for non-benefited individuals shifts the recommended allocation from broadly assigning \(A=1\) to a more selective policy.
Because $C_u$ represents a utility trade-off, this sensitivity highlights the importance of eliciting decision makers' relative valuations of benefit and non-benefit before choosing an allocation rule.
The proposed method makes this trade-off explicit while accommodating ordinal health outcomes without assigning arbitrary numerical distances between health categories and accounting for partial identification of individual benefit.
Within this framework, the orthogonal smoothed estimator produces a more gradual allocation response than the plug-in estimators around the transition, reducing the abruptness of changes in the average recommended treated proportion as preferences vary.
These features provide a practical way to translate stated preferences into allocation rules and examine how those rules change under alternative utility judgments.

\section{Discussion}\label{sec:discuss}

This paper develops a theory of minimax regret policy learning for ordinal outcomes with heterogeneous utilities and partially identified probability of individual benefit. We characterize the optimal treatment rule and establish excess regret guarantees that quantify how nonsmooth identification bounds affect policy learning. These guarantees reveal how the margin condition, nuisance estimation rates, and policy class complexity govern the choice of smoothing parameter and convergence toward the best policy in the chosen class.
{Our construction builds on smoothing-based estimation of causal bounds \citep{levis2025covariate} and irregular functionals \citep{whitehouse2025inference}. Here, smoothing and orthogonalization are used to learn a minimax regret policy for ordinal outcomes with subgroup-specific utilities. The analysis allows $\beta_n$ to diverge, explicitly balances smoothing approximation error against the nuisance remainder that grows with $\beta_n$, and incorporates policy class complexity into the excess-regret guarantee.}

Several extensions merit further study. Bounds under alternative identification assumptions, including settings with unmeasured confounding or noncompliance \citep{gabriel2024sharp}, could motivate related decision criteria, but their incorporation requires corresponding score constructions and regularity analysis. Allowing stochastic policies may improve minimax regret in ambiguous regions \citep{cui2021individualized,cui2025policy}. 
Extensions to multiple treatments or continuous outcomes could draw on
multi-action policy learning \citep{zhou2023offline} and bounds on the
distribution of treatment effects \citep{fan2010sharp}, respectively,
while requiring corresponding developments in the utility specification,
identification bounds, and estimation theory.

\section*{Supplementary Material}

{The Supplementary Material provides proofs, additional theoretical results, simulation results, and implementation details.}

Code for reproducing the simulations and the SIPP analysis is available at \url{https://github.com/tsanyuih/orthogonal_PL}.

\bibliographystyle{apalike}
\bibliography{mybib}

\clearpage
\appendix
\setcounter{equation}{0}
\renewcommand{\theequation}{S\arabic{equation}}
\setcounter{lemma}{0}
\renewcommand{\thelemma}{S\arabic{lemma}}
\setcounter{assumption}{0}
\renewcommand{\theassumption}{S\arabic{assumption}}
\setcounter{theorem}{0}
\renewcommand{\thetheorem}{S\arabic{theorem}}
\setcounter{proposition}{0}
\renewcommand{\theproposition}{S\arabic{proposition}}

\setcounter{table}{0}
\renewcommand{\thetable}{S\arabic{table}}
\setcounter{figure}{0}
\renewcommand{\thefigure}{S\arabic{figure}}

\renewcommand{\theHequation}{S\arabic{equation}}
\renewcommand{\theHtable}{S\arabic{table}}
\renewcommand{\theHfigure}{S\arabic{figure}}
\renewcommand{\theHlemma}{S\arabic{lemma}}
\renewcommand{\theHassumption}{S\arabic{assumption}}
\renewcommand{\theHtheorem}{S\arabic{theorem}}
\renewcommand{\theHproposition}{S\arabic{proposition}}

\begin{center}
\Large \bf Supplementary Material
\end{center}
\begin{appendices}
The Supplementary Material provides proofs, additional theoretical results, simulation results, and implementation details.

\section{Properties of the smoothed maximum function}\label{app:sm-property}
Given a vector of functions $ g := ( g_1, \ldots, g_K ) $ with $ g_k: \mathcal{X} \rightarrow \mathbb{R}, ~~ k = 1, \ldots, K $.
Fix $ x \in \mathcal{X} $, $ g \in \mathbb{R}^K $ is a constant vector.
Our target function is $ G:= \max_{1 \leq k \leq K} g_k $.
We now prove a variety of properties for the smoothed maximum approximation of $G$,
\begin{equation*}
G^{\beta_n} := \frac{ \sum_{k = 1}^K g_k \exp (\beta_n g_k) }{ \sum_{i = 1}^K \exp(\beta_n g_i ) }.
\end{equation*}
We start by proving a number of basic analytical properties of the soft maximum function that will be useful throughout our work.
In the sequel, we define the $k$-th softmax weight $ w_k^{\beta_n} (g) $ to be
\begin{equation*}
w_k^{\beta_n} (g) := \frac{ \exp (\beta_n g_k) }{ \sum_{i = 1}^K \exp(\beta_n g_i ) }.
\end{equation*}

{In the following lemma, $\|\cdot\|_2$ and $\|\cdot\|_\infty$ denote the Euclidean and maximum norms of vectors, and $\|\cdot\|_{op}$ denotes the matrix operator norm induced by $\|\cdot\|_2$. We write $a\lesssim b$ if $a\leq Cb$ for some constant $C>0$.}

\begin{lemma}\label{lem:sm-property}
For any $\beta_n > 0$, the following properties hold for the smoothed maximum approximation $ G^{\beta_n} $.
\begin{enumerate}[label = (\roman*)]
\item If there exist $i, j \in \{ 1, \ldots, K \}$ such that $ g_i \neq g_j $, then $ G^{\beta_n} $ is strictly increasing with respect to $\beta_n$.
Furthermore,
\begin{equation*}
0 \leq \max_{1 \leq k \leq K} g_k - G^{\beta_n} \leq \frac{\log K}{\beta_n}.
\end{equation*}
As a result, $ G^{\beta_n} \rightarrow \max_{1 \leq k \leq K} g_k $ as $\beta_n \rightarrow \infty$.

\item The partial derivative of $ G^{\beta_n} $ with respect to the $k$-th component is given by
\begin{equation*}
\frac{\partial G^{\beta_n}}{\partial g_k} = w_k^{\beta_n} (g) \left\{ 1 + \beta_n \cdot \big( g_k - G^{\beta_n} \big) \right\}.
\end{equation*}
Furthermore, we have $$ \sup_{g \in \mathbb{R}^K} || \nabla_g G^{\beta_n} ||_{\infty} \leq 1 + \log K.
$$

\item The function $ G^{\beta_n} $ is Lipschitz continuous, that is,
\begin{equation*}
| G^{\beta_n} - H^{\beta_n} | \lesssim || g - h ||_{2},
\end{equation*}
where $ h \in \mathbb{R}^K $ and $H^{\beta_n}$ is the smoothed maximum approximation of $h$.

\item The Hessian matrix of $ G^{\beta_n} $ satisfies $ || \nabla_g^2 G^{\beta_n} ||_{op} \lesssim \beta_n $.

\end{enumerate}
\end{lemma}

\begin{proof}

Write $\beta=\beta_n$, $w_k=w_k^{\beta}(g)$, and $\mathcal E(w)=-\sum_{k=1}^K w_k\log w_k$. The weights are positive and sum to one, and $0\leq \mathcal E(w)\leq\log K$.

\noindent (i) Direct differentiation gives
\begin{equation*}
\frac{\partial G^{\beta}}{\partial\beta}
=\sum_{k=1}^K w_k(g_k-G^{\beta})^2\geq0.
\end{equation*}
The derivative is strictly positive unless all components of $g$ are equal. Moreover, the identity
\begin{equation*}
\log\Big(\sum_{k=1}^K e^{\beta g_k}\Big)=\beta G^{\beta}+\mathcal E(w)
\end{equation*}
and $\log(\sum_k e^{\beta g_k})\geq\beta G$ imply
\begin{equation*}
0\leq G-G^{\beta}\leq \mathcal E(w)/\beta\leq(\log K)/\beta.
\end{equation*}
The lower bound follows because $G^{\beta}$ is a weighted average of the components. This also proves convergence as $\beta\to\infty$.

\noindent (ii) Since $\partial w_k/\partial g_l=\beta w_k(\delta_{kl}-w_l)$, where $\delta_{kl}$ is the Kronecker delta,
\[
q_k:=\frac{\partial G^{\beta}}{\partial g_k}
=w_k\{1+\beta(g_k-G^{\beta})\}
=w_k(1+\log w_k+\mathcal E(w)).
\]
Consequently, $q_k\leq1+\log K$, while $q_k\geq w_k\log w_k\geq-1/e$. Thus $\max_k|q_k|\leq1+\log K$, uniformly in $g$ and $\beta$.

\noindent (iii) The mean value theorem and (ii) yield
\[
|G^{\beta}-H^{\beta}|
\leq\sqrt K(1+\log K)\|g-h\|_2.
\]

\noindent (iv) Differentiating $q_k$ gives
\[
\frac{\partial q_k}{\partial g_l}
=\beta\{(\delta_{kl}-w_l)q_k+w_k(\delta_{kl}-q_l)\}.
\]
Therefore each row of the Hessian satisfies
\begin{equation*}
\begin{aligned}
\sum_{l=1}^K\left|\frac{\partial q_k}{\partial g_l}\right|&\leq\beta\left\{2|q_k|+w_k\left(1+\sum_{l=1}^K|q_l|\right)\right\}\\ &\leq\beta\{2(1+\log K)+1+K(1+\log K)\}.
\end{aligned}
\end{equation*}
Since the Hessian is symmetric, its operator norm is at most its maximum absolute row sum. This proves the claimed bound, with constants depending only on the fixed dimension $K$.
\end{proof}

{For a vector $v=(v_1,\ldots,v_K)\in\mathbb R^K$, let $v_{(1)}\geq v_{(2)}\geq\cdots\geq v_{(K)}$ denote its ordered components, as in Assumption~\ref{assump:margin-cond}. Define the component gaps and the smoothing error, respectively, by
\[
\Delta_i(v):=v_{(1)}-v_i\geq0,
\qquad
\Delta^{\beta_n}(v):=v_{(1)}-G^{\beta_n}(v)\geq0.
\]
These definitions apply pointwise when $v$ is random. The subscript $i$ indexes a component, while the superscript $\beta_n$ specifies the smoothing parameter. We bound the expected smoothing error by adapting the softmax bias argument of \cite{chen2023inference}.}

\begin{lemma}\label{lem:approx-error}
{Let $v=(v_1,\ldots,v_K)$ be a random vector and $\beta_n\geq0$. Then
\[
\EP[\Delta^{\beta_n}(v)]
\leq\sum_{i=1}^K\EP\left[\Delta_i(v)e^{-\beta_n\Delta_i(v)}\right].
\]
Furthermore, if the top-two gap satisfies
$\prob(v_{(1)}-v_{(2)}\leq t)\leq C_0t^\delta$ for $t\geq0$ and some $\delta>0$, as in Assumption~\ref{assump:margin-cond}, there is a constant $C>0$ depending only on $C_0$ and $\delta$ such that
\[
\EP\left[\Delta_i(v)e^{-\beta_n\Delta_i(v)}\right]
\leq\frac{C}{\beta_n^{1+\delta}},
\qquad i=1,\ldots,K,\quad\beta_n>0.
\]
Consequently, $\EP[\Delta^{\beta_n}(v)]\leq KC\beta_n^{-(1+\delta)}$.}
\end{lemma}

\begin{proof}

Write $\beta=\beta_n$ and, within this proof, abbreviate $\Delta_i=\Delta_i(v)$. Since at least one $\Delta_i$ equals zero,
\[
0\leq\Delta^\beta(v)
=\frac{\sum_{i=1}^K\Delta_i e^{-\beta\Delta_i}}{\sum_{i=1}^K e^{-\beta\Delta_i}}
\leq\sum_{i=1}^K\Delta_i e^{-\beta\Delta_i}.
\]
Taking expectations proves the first bound. Whenever $\Delta_i>0$, we have $v_{(1)}-v_{(2)}\leq\Delta_i$, so
\[
F_i(t):=\prob(0<\Delta_i\leq t)
\leq\prob(v_{(1)}-v_{(2)}\leq t)\leq C_0t^\delta.
\]
For $\beta>0$, partitioning $(0,\infty)$ into intervals $(r/\beta,(r+1)/\beta]$, $r=0,1,\ldots$, gives
\[
\begin{aligned}
\EP[\Delta_i e^{-\beta\Delta_i}]\leq\frac{1}{\beta}\sum_{r=0}^\infty(r+1)e^{-r}F_i((r+1)/\beta)\leq\frac{C_0}{\beta^{1+\delta}}\sum_{r=0}^\infty(r+1)^{1+\delta}e^{-r}.
\end{aligned}
\]
The series is finite for every $\delta>0$, so the claimed constant depends only on $C_0$ and $\delta$. Summing over $i$ gives the bound for $\EP[\Delta^\beta(v)]$.
\end{proof}

\section{Proposition proofs and orthogonality background}

\subsection{Proof of Proposition~\ref{prop:max-regret}}
\begin{proof}

For $\pi(X)\in\{0,1\}$, the conditional regret at $X$ is
\[
\begin{aligned}
\{u^b(X)-u^n(X)\}\Big[&(1-\pi(X))\max\{\eta(X)-C_u(X),0\}\\ &+\pi(X)\max\{C_u(X)-\eta(X),0\}\Big].
\end{aligned}
\]
Since $u^b(X)-u^n(X)>0$, its maximum over $\eta(X)\in[\eta_L(X),\eta_U(X)]$ is attained at
$\bar\eta(X)=\pi(X)\eta_L(X)+(1-\pi(X))\eta_U(X)$.
This measurable endpoint selection simultaneously attains the pointwise maxima, so
\begin{align*}
R_{\sup}(\pi)
&=\EP\Big[\{u^b(X)-u^n(X)\}\{(1-\pi(X))\psi_U(X)-\pi(X)\psi_L(X)\}\Big]\\
&=-\EP\Big[\{u^b(X)-u^n(X)\}\pi(X)\{\psi_U(X)+\psi_L(X)\}\Big]\\ &\quad+\EP\Big[\{u^b(X)-u^n(X)\}\psi_U(X)\Big],
\end{align*}
where $\psi_L(X)=\min\{\eta_L(X)-C_u(X),0\}$ and
$\psi_U(X)=\max\{\eta_U(X)-C_u(X),0\}$.
This gives Proposition~\ref{prop:max-regret}, with
$C=\EP[\{u^b(X)-u^n(X)\}\psi_U(X)]$, which is independent of $\pi$.
\end{proof}

\subsection{Proof of Proposition~\ref{prop:approx-error}}
\begin{proof}

Using the smoothing error $\Delta^{\beta_n}(v)$ defined before Lemma~\ref{lem:approx-error}, we obtain the following identity.
Suppressing the argument $X$, the definitions of the four smoothed maxima give
\begin{equation}\label{eq:smoothing-error-identity}
\begin{aligned}
\psi^{\beta_n}-\psi
=-\Delta^{\beta_n}(\delta_L)+\Delta^{\beta_n}((\delta_L,0))+\Delta^{\beta_n}(-\delta_U)-\Delta^{\beta_n}((-\delta_U,0)).
\end{aligned}
\end{equation}
With the common additive constant omitted as in Section~\ref{sec:smooth-approx}, the regret difference is therefore
\[
R_{\sup}^{\beta_n}(\pi)-R_{\sup}(\pi)=-\EP\Big[\{u^b(X)-u^n(X)\}\pi(X)\{\psi^{\beta_n}(X)-\psi(X)\}\Big].
\]
By Lemma~\ref{lem:sm-property}, each $\Delta^{\beta_n}(v)$ is between zero and $\log(\dim v)/\beta_n$.
The triangle inequality, boundedness of $u^b-u^n$, and $0\leq\pi\leq1$ yield
$|R_{\sup}^{\beta_n}(\pi)-R_{\sup}(\pi)|\lesssim\beta_n^{-1}$.
\end{proof}

\subsection{Orthogonality background and proof of Proposition~\ref{prop:orthogonal-score}}\label{app:orthog-background}

{
\noindent\textbf{Plug-in expansion and Neyman orthogonality.}\par
For a fixed policy $\pi$ and smoothing parameter $\beta_n>0$, write
\[
\Gamma^{\beta_n}(O,\pi;\theta)
=-\{u^b(X)-u^n(X)\}\pi(X)\psi^{\beta_n}(X;m),
\]
so that $R_{\sup}^{\beta_n}(\pi;\theta)=\EP[\Gamma^{\beta_n}(O,\pi;\theta)]$, with the policy-independent constant omitted as in Section~\ref{sec:smooth-approx} of the main paper. To describe the effect of nuisance estimation, consider the expansion \citep{Kennedy2023Semiparametric, chernozhukov2022locally}
\begin{equation}\label{eq:von-mises}
\begin{aligned}
&R_{\sup}^{\beta_n}(\pi;\hat\theta)-R_{\sup}^{\beta_n}(\pi;\theta_0)=\partial_\theta\EP[\Gamma^{\beta_n}(O,\pi;\theta_0)][\hat\theta-\theta_0]
+R_2(\hat\theta,\theta_0).
\end{aligned}
\end{equation}
Here $R_2$ is the second-order remainder, whose dependence on $\pi$ and $\beta_n$ is suppressed. Its bound can grow with $\beta_n$; the gradient and Hessian bounds in Lemma~\ref{lem:sm-property} describe this dependence. The first term is the G\^ateaux derivative with respect to the nuisance functions, with the expectation taken under the true distribution. For an admissible direction $h=\theta-\theta_0$, it is defined by
\[
\begin{aligned}
&\partial_\theta\EP[\Gamma^{\beta_n}(O,\pi;\theta_0)][h] =\lim_{t\to0}\frac{\EP[\Gamma^{\beta_n}(O,\pi;\theta_0+t h)]
-\EP[\Gamma^{\beta_n}(O,\pi;\theta_0)]}{t},
\end{aligned}
\]
whenever the limit exists. The admissible paths used for the score verification are specified below.

\begin{definition}\label{def:orthogonal-score}
A function $\Psi^{\beta_n}:\mathcal O\times\Pi\times\Theta\to\mathbb R$ is a Neyman orthogonal score for $R_{\sup}^{\beta_n}$ if, for every fixed $\pi\in\Pi$, it satisfies:
\begin{enumerate}[label=(\alph*)]
\item Its expectation recovers the target at the true nuisance functions:
\[
\EP[\Psi^{\beta_n}(O,\pi;\theta_0)]=R_{\sup}^{\beta_n}(\pi;\theta_0).
\]
\item Its expected value has zero derivative at $\theta_0$ along every admissible nuisance direction $h$:
\begin{equation}\label{eq:NO-cond}
\partial_\theta\EP[\Psi^{\beta_n}(O,\pi;\theta_0)][h]=0.
\end{equation}
\end{enumerate}
\end{definition}
Condition~\eqref{eq:NO-cond} removes the first-order effect of nuisance estimation on the expected score. Under suitable smoothness conditions, the remaining nuisance error is second-order, although its magnitude can depend on $\beta_n$. Theorem~\ref{thm:asymp-converg} quantifies this remainder together with the smoothing approximation error.

\medskip
\noindent\textbf{Verification of the proposed score.}\par
}

We prove the Neyman orthogonality of $ \Psi_L^{\beta_n} $ with respect to the generic nuisance functions $e$ and $m_k$ ($k \in \mathcal{J}$), evaluated at their true values $e_0$ and $m_{0,k}$; a similar proof applies to $ \Psi_U^{\beta_n} $.
Define
\begin{align*}
\phi_L^{\beta_n} (O; \theta) &= \sum_{j = 0}^{J - 1} \frac{ \partial \psi_L^{\beta_n} (X; {m}) }{ \partial \delta_{L, j} } L_j (O; \theta), \\ \phi_U^{\beta_n} (O; \theta) &= \sum_{j = 0}^{J - 1} \frac{ \partial \psi_U^{\beta_n} (X; {m}) }{ \partial \delta_{U, j} } U_j (O; \theta).
\end{align*}
{The full score is $\Psi^{\beta_n}(O,\pi;\theta)=\allowbreak -\{u^b(X)-u^n(X)\}\pi(X)\{\Psi_L^{\beta_n}(O;\theta)+\allowbreak \Psi_U^{\beta_n}(O;\theta)\}$. It suffices to verify the conditional derivative cancellations for the two components, since the prefactor depends only on $X$.} Here
\begin{align*}
\Psi_L^{\beta_n} (O; \theta) = \psi_L^{\beta_n} (X; {m}) + \phi_L^{\beta_n} (O; \theta), ~~ \Psi_U^{\beta_n} (O; \theta) = \psi_U^{\beta_n} (X; {m}) + \phi_U^{\beta_n} (O; \theta).
\end{align*}

At $\theta_0$, $\EP[T_{a,j}(O;\theta_0)\mid X]=0$ for every $a$ and $j$. The definitions of $L_j$, $U_j$, and $\phi^{\beta_n}$ therefore imply $\EP[\phi^{\beta_n}(O;\theta_0)\mid X]=0$. Consequently,
\[
\begin{aligned}
\EP[\Psi^{\beta_n}(O,\pi;\theta_0)]&=-\EP\Big[\{u^b(X)-u^n(X)\}\pi(X)\psi^{\beta_n}(X;m_0)\Big]\\ &=R_{\sup}^{\beta_n}(\pi;\theta_0).
\end{aligned}
\]

For fixed $\beta_n>0$, the derivatives below are taken along admissible nuisance paths
$m_t=m_0+t h_m$ and $e_t=e_0+t h_e$.
The treatment-specific outcome directions $h_{m,j}(a,\cdot)$ are bounded, and the propensity path remains in $(0,1)$ with
\[
\sup_{|t|\leq t_0}\max_{a\in\{0,1\}}
\left\|\frac{e_0(a,X)}{e_t(a,X)}\right\|_\infty<\infty
\]
for some $t_0>0$, where $e_t(1,X)=e_t(X)$ and $e_t(0,X)=1-e_t(X)$.

We first check orthogonality with respect to $m_k$, evaluated at $m_{0,k}$.
Note that $ \Psi_L^{\beta_n} $ is now pathwise differentiable with smoothed functional $ \psi^{\beta_n} $ instead of the nonsmooth functional $ \psi $.
We can check that the Gateaux derivative of $ \EP [ \Psi_L^{\beta_n} (O; \theta) ] $ with respect to each $ m_{k} $ vanishes, evaluated at $ \theta_0 $, that is,
\begin{equation}\label{eq:orthog-outcome}
\partial_{m_{k}} \EP [ \Psi_L^{\beta_n} (O; \theta_0) ] [\tilde{m}_{k} - m_{0, k}] = 0,
\end{equation}
{Here $\tilde m_k-m_{0,k}$ is an admissible outcome direction as specified above; in particular, each treatment-specific direction is square-integrable under the marginal law of $X$.}
Let $h_k(a,X):=\tilde{m}_k(a,X)-m_{0,k}(a,X)$. The scalar partial derivatives of $\delta_{L,j}$ with respect to the relevant components of $m_k$ are $-1$, 0 or 1, while the corresponding Gateaux derivatives in direction $h_k$ include the factor $h_k(a,X)$.
Hence
\begin{align*}
& \partial_{m_{k}} \EP [ \psi_L^{\beta_n} (X; m_0) ] [\tilde{m}_{k} - m_{0, k}] = \EP \Big[ \partial_{m_{k}} [ \psi_L^{\beta_n} (X; m_0) ] [\tilde{m}_{k} - m_{0, k}] \Big] \\
= \ & \EP \left[ \sum_{j = 0}^{J - 1} \frac{\partial \psi_L^{\beta_n} (X; m_0)}{\partial \delta_{L, j}} \partial_{m_{k}} [ \delta_{L, j} (X; m_0) ] [\tilde{m}_{k} - m_{0, k}] \right] \\
= \ & \EP \left[ \sum_{j = 0}^{J - 1} \frac{{A} \cdot \indicator ( j \leq k \leq J - 1 )}{e_0(X)} h_k(1,X) \frac{\partial \psi_L^{\beta_n} (X; m_0)}{\partial \delta_{L, j}} \right] \\
& - \EP \left[ \sum_{j = 0}^{J - 1} \frac{(1 - {A}) \cdot \indicator ( j \leq k \leq J - 1 )}{1 - e_0(X)} h_k(0,X) \frac{\partial \psi_L^{\beta_n} (X; m_0)}{\partial \delta_{L, j}} \right] \\
= \ & \EP \left[ \sum_{j = 0}^{k} \frac{{A}}{e_0(X)} h_k(1,X) \frac{\partial \psi_L^{\beta_n} (X; m_0)}{\partial \delta_{L, j}} \right] \\ & - \EP \left[ \sum_{j = 0}^{k} \frac{1 - {A}}{1 - e_0(X)} h_k(0,X) \frac{\partial \psi_L^{\beta_n} (X; m_0)}{\partial \delta_{L, j}} \right],
\end{align*}
{The first equality follows by dominated differentiation, using the uniform gradient bound in Lemma~\ref{lem:sm-property} and the integrable direction $h_k$; the second equality follows from the chain rule.}
Moreover, note that we can rewrite $ \phi_L^{\beta_n} $ as
\begin{align*}
\phi_L^{\beta_n} (O; \theta_0) = \ & \sum_{j = 0}^{J - 1} \frac{ \partial \psi_L^{\beta_n} (X; m_0) }{ \partial \delta_{L, j} } L_j (O;\theta_0) \\
= \ & \frac{ {A} }{ e_0(X) } \sum_{j = 0}^{J - 1} \frac{\partial \psi_L^{\beta_n} (X; m_0)}{\partial \delta_{L, j}} \sum_{k = j}^{J - 1} \Big( \indicator \{ {Y} = k \} - m_{0, k} (A, X) \Big) \\
& - \frac{ 1 - {A} }{ 1 - e_0(X) } \sum_{j = 0}^{J - 1} \frac{\partial \psi_L^{\beta_n} (X; m_0)}{\partial \delta_{L, j}} \sum_{k = j}^{J - 1} \Big( \indicator \{ {Y} = k \} - m_{0, k} (A, X) \Big) \\
= \ & \frac{ {A} }{ e_0(X) } \sum_{k = 0}^{J - 1} \Big( \indicator \{ {Y} = k \} - m_{0, k} (A, X) \Big) \sum_{j = 0}^{k} \frac{\partial \psi_L^{\beta_n} (X; m_0)}{\partial \delta_{L, j}} \\
& - \frac{ 1 - {A} }{ 1 - e_0(X) } \sum_{k = 0}^{J - 1} \Big( \indicator \{ {Y} = k \} - m_{0, k} (A, X) \Big) \sum_{j = 0}^{k} \frac{\partial \psi_L^{\beta_n} (X; m_0)}{\partial \delta_{L, j}},
\end{align*}
where $ m_{0, k} (a, X) = \prob ( Y_a = k \mid X ) $.
Then we have
\begin{align*}
& \partial_{m_{k}} \EP [ \phi_L^{\beta_n} (O; \theta_0) ] [\tilde{m}_{k} - m_{0, k}] \\ ={}& - \EP \left[ \sum_{j = 0}^{k} \frac{{A}}{e_0(X)} h_k(1,X) \frac{\partial \psi_L^{\beta_n} (X; m_0)}{\partial \delta_{L, j}} \right] + \EP \left[ \sum_{j = 0}^{k} \frac{1 - {A}}{1 - e_0(X)} h_k(0,X) \frac{\partial \psi_L^{\beta_n} (X; m_0)}{\partial \delta_{L, j}} \right].
\end{align*}
{The additional terms obtained by differentiating the gradient weights have conditional mean zero, because $\EP\{L_j(O;\theta_0)\mid X\}=0$.} This proves \eqref{eq:orthog-outcome}.

{For the propensity direction, hold $m=m_0$ fixed and write
$D_{L,j}(X)=\partial\psi_L^{\beta_n}(X;m_0)/\partial\delta_{L,j}$.
For every sufficiently small $t$, the exact conditional expectation is
\begin{align*}
&\EP[\phi_L^{\beta_n}(O;e_t,m_0)\mid X]\\
={}&\sum_{j=0}^{J-1}D_{L,j}(X)\sum_{k=j}^{J-1}
\Bigg\{\frac{e_0(X)}{e_t(X)}
\EP[\indicator\{Y=k\}-m_{0,k}(1,X)\mid A=1,X]\\
&\qquad-\frac{1-e_0(X)}{1-e_t(X)}
\EP[\indicator\{Y=k\}-m_{0,k}(0,X)\mid A=0,X]\Bigg\} =0.
\end{align*}
The path conditions make the correction integrable, so its expectation is identically zero along the path. Differentiating this constant expectation gives
$\partial_e\EP[\phi_L^{\beta_n}(O;\theta_0)][h_e]=0$;
$\psi_L^{\beta_n}(X;m_0)$ does not depend on $e$.
The same conditional cancellations apply to the upper bound component.
For simultaneous admissible perturbations of $e$ and $m$, the identity
\[
\EP[T_{a,j}(O;e_t,m_t)\mid X]
=-t\,\frac{e_0(a,X)}{e_t(a,X)}h_{m,j}(a,X)
\]
shows that the correction cancels the derivative of the smoothed plug-in term at $t=0$.
The bounded path directions, propensity ratios, and gradient and Hessian bounds in Lemma~\ref{lem:sm-property} justify dominated differentiation for fixed $\beta_n$.
Multiplication by the bounded factor $-\{u^b(X)-u^n(X)\}\pi(X)$ therefore gives the required zero derivative for the full score.}
This completes the proof.

\section{Additional theorem and proofs}\label{app:pf-thm}

\subsection{Fixed-policy regret estimation}\label{app:fixed-policy}

The following expansion separates sampling variability, smoothing approximation error, and nuisance estimation error for a fixed policy. It supports the excess-regret guarantee in Theorem~\ref{cor:regret-converge} of the main paper.
For a function $g$ of $O$, we use $\|g\|^2=\EP[\|g(O)\|_2^2]$ for its squared $L_2(\mathbb{P})$ norm, where $\|\cdot\|_2$ is the Euclidean norm.

\begin{theorem}\label{thm:asymp-converg}
Suppose Assumptions~\ref{assump:unconf-pos} -~\ref{assump:margin-cond} hold.
Moreover, we have access to $L^2$-consistent estimators of the propensity score and conditional outcome probabilities,
\begin{equation*}
|| \hat{e} - e_0 || = o_{\mathbb{P}} (1), ~~ || \hat{m}_{j} - m_{0, j} || = o_{\mathbb{P}} (1), ~~ j \in \mathcal{J},
\end{equation*}
where $ \hat{m}_{j} $ and $ \hat{e} $ are estimated from an independent sample.
Suppose $ || \hat{m} (A, X) ||_{L_{\infty} (\mathbb{P})} < B $ for some constant $B > 0$, and $ \prob (\epsilon \leq \hat{e} (X) \leq 1 - \epsilon) = 1 $ for some $ \epsilon > 0 $.
For $\beta_n\to\infty$, given a policy $\pi \in \Pi_n$, we have
\begin{equation}\label{eq:thm1}
\begin{aligned}
& \hat{R}_{\sup}^{\beta_n, \no} (\pi; \hat{\theta}) - R_{\sup} (\pi; \theta_0) \\
&\quad= (\Pn - \EP) \Psi^* (O, \pi; \theta_0) \\
&\qquad + O_{\mathbb{P}} \bigg( || \hat{e} - e_0 || \cdot \max_{0 \leq j \leq J - 1} || \hat{m}_{j} - m_{0, j} || \bigg) \\
&\qquad + O_{\mathbb{P}} \left( \beta_n^{- (1 + \delta)} \right)
+ O_{\mathbb{P}} \bigg( \beta_n \max_{0 \leq j \leq J - 1} || \hat{m}_{j} - m_{0, j} ||^2 \bigg) + o_{\mathbb{P}} (n^{-1 / 2}),
\end{aligned}
\end{equation}
where $ \Psi^* $ is the probability limit of $ \Psi^{\beta_n} $ defined in \eqref{eq:NO-score} as $ \beta_n \rightarrow \infty $.
\end{theorem}

\begin{remark}[Sufficient conditions for root-$n$ consistency]
\label{rem:alternative-rates}
Under the conditions of Theorem~\ref{thm:asymp-converg},
regret estimation for a fixed policy is root-$n$ consistent
if $\beta_n \gtrsim n^{1/\{2(1+\delta)\}}$ and
\[
\|\hat e-e_0\|
\max_{j\in\mathcal J}\|\hat m_j-m_{0,j}\|
=o_{\mathbb P}(n^{-1/2}),
~~\max_{j\in\mathcal J}\|\hat m_j-m_{0,j}\|
=o_{\mathbb P}(n^{-1/4}\beta_n^{-1/2}).
\]
These conditions make the smoothing approximation error $O(n^{-1/2})$ and both nuisance estimation remainders $o_{\mathbb P}(n^{-1/2})$. They do not require the nuisance estimators to converge at polynomial rates.
\end{remark}

\medskip
\noindent\textbf{Proof of Theorem~\ref{thm:asymp-converg}.}\par

First, we do the following decomposition:
\begin{equation}\label{eq:regret-decomp}
\begin{aligned}
&\hat{R}_{\sup}^{\beta_n, \no} (\pi; \hat{\theta}) - R_{\sup} (\pi; \theta_0) \\ ={}& \left( \hat{R}_{\sup}^{\beta_n, \no} (\pi; \hat{\theta}) - R_{\sup}^{\beta_n, \no} (\pi; \theta_0) \right) + \left( R_{\sup}^{\beta_n, \no} (\pi; \theta_0) - R_{\sup} (\pi; \theta_0) \right).
\end{aligned}
\end{equation}
The second term in \eqref{eq:regret-decomp} is the approximation error via smoothing, which can be controlled by the growing rate of $ \beta_n $.
{At the true nuisance functions, $\EP[\phi^{\beta_n}(O;\theta_0)\mid X]=0$. Hence
\begin{equation}\label{eq:noone-6}
\begin{aligned}
&R_{\sup}^{\beta_n,\no}(\pi;\theta_0)-R_{\sup}(\pi;\theta_0)\\
=\ & -\EP\Big[\{u^b(X)-u^n(X)\}\pi(X)\{\psi^{\beta_n}(X;m_0)-\psi(X;m_0)\}\Big].
\end{aligned}
\end{equation}
Applying the triangle inequality to \eqref{eq:smoothing-error-identity} and using $0\leq\pi\leq1$, the absolute difference is bounded by $\|u^b-u^n\|_\infty$ times the sum of the expected approximation errors $\Delta^{\beta_n}(v)$ for the four vectors
\[
v\in\{\delta_L,(\delta_L,0),-\delta_U,(-\delta_U,0)\}.
\]
Assumption~\ref{assump:margin-cond} and Lemma~\ref{lem:approx-error} bound each expectation by a constant times $\beta_n^{-(1+\delta)}$. Boundedness of $u^b-u^n$ therefore gives the same rate for the weighted regret difference.}

Now we decompose $ \hat{R}_{\sup}^{\beta_n, \no} (\pi; \hat{\theta}) - R_{\sup}^{\beta_n, \no} (\pi; \theta_0) $ as
\begin{equation}\label{eq:regret-decomp-1}
\begin{aligned}
& \hat{R}_{\sup}^{\beta_n, \no} (\pi; \hat{\theta}) - R_{\sup}^{\beta_n, \no} (\pi; \theta_0) \\
= \ & \Pn \Psi^{\beta_n} (O, \pi; \hat{\theta}) - \EP \Psi^{\beta_n} (O, \pi; \theta_0) \\
= \ & (\Pn - \EP) \{ \Psi^{\beta_n} (O, \pi; \hat{\theta}) - \Psi^{\beta_n} (O, \pi; \theta_0) \} + (\Pn - \EP) \{ \Psi^{\beta_n} (O, \pi; \theta_0) \} \\
& + \EP \{ \Psi^{\beta_n} (O, \pi; \hat{\theta}) - \Psi^{\beta_n} (O, \pi; \theta_0) \} \\
:= \ & \text{(I)} + \text{(II)} + \text{(III)}.
\end{aligned}
\end{equation}
Next we show that the $ \text{(I)} = o_{\mathbb{P}} (n^{-1 / 2}) $, Term (II) converges to $ (\Pn - \EP) \{ \Psi^* (O, \pi; \theta_0) \} $, and
\begin{equation*}
\text{(III)} = O_{\mathbb{P}} \bigg( || \hat{e} - e_0 || \cdot \max_{0 \leq j \leq J - 1} || \hat{m}_{j} - m_{0, j} || + \beta_n \max_{0 \leq j \leq J - 1} || \hat{m}_{j} - m_{0, j} ||^2 \bigg)
\end{equation*}
under the conditions of Theorem~\ref{thm:asymp-converg}.

\medskip
\noindent\textbf{Term (I).}
By Lemma 2 in \cite{kennedy2020sharp}, the first term in \eqref{eq:regret-decomp-1} is $ o_{\mathbb{P}} (n^{-1 / 2}) $ if
\begin{equation*}
\EP \left[ \{ \Psi^{\beta_n} (O, \pi; \hat{\theta}) - \Psi^{\beta_n} (O, \pi; \theta_0) \}^2 \right] = o_{\mathbb{P}} (1).
\end{equation*}
In fact,
\begin{align}
& \EP \left[ \left\{ \Psi^{\beta_n} (O, \pi; \hat{\theta}) - \Psi^{\beta_n} (O, \pi; \theta_0) \right\}^2 \right] = \| \Psi^{\beta_n} (O, \pi; \hat{\theta}) - \Psi^{\beta_n} (O, \pi; \theta_0) \|^2 \notag \\
{\lesssim} \ & 2 \left\{ \| \psi^{\beta_n} (X; \hat{m}) - \psi^{\beta_n} (X; m_0) \|^2 + \| \phi^{\beta_n} (O; \hat{\theta}) - \phi^{\beta_n} (O; \theta_0) \|^2 \right\}. \label{eq:Term1-1}
\end{align}
The first term of \eqref{eq:Term1-1} can be upper bounded using the {Lipschitz} property of smoothed maximum function,
{
\begin{align*}
& \|\psi^{\beta_n}(X;\hat m)-\psi^{\beta_n}(X;m_0)\|^2\\
\lesssim{}&\EP\Big\{\|\delta_L(X;\hat m)-\delta_L(X;m_0)\|_2^2
+\|\delta_U(X;\hat m)-\delta_U(X;m_0)\|_2^2\Big\}\\
={}&\sum_{j=0}^{J-1}\Big\{\|\delta_{L,j}(X;\hat m)-\delta_{L,j}(X;m_0)\|^2
+\|\delta_{U,j}(X;\hat m)-\delta_{U,j}(X;m_0)\|^2\Big\}\\
\lesssim{}&\sum_{j=0}^{J-1}\|\hat m_j-m_{0,j}\|^2=o_{\mathbb P}(1).
\end{align*}
The first inequality uses the Lipschitz bound in Lemma~\ref{lem:sm-property}, and the equality expands the squared Euclidean norms.
The final inequality follows from the affine definitions of the bound components and Cauchy--Schwarz for finite sums, using positivity of $e_0$ to compare treatment-specific and observed-data norms.
The $o_{\mathbb P}(1)$ conclusion follows from nuisance consistency and fixed $J$.}
{For the second term, write $\widehat D_{L,j}=\partial\psi_L^{\beta_n}(X;\hat m)/\partial\delta_{L,j}$ and $D_{L,j,0}=\partial\psi_L^{\beta_n}(X;m_0)/\partial\delta_{L,j}$, with analogous notation for the upper bound component. The required decomposition is
\begin{align*}
&\widehat D_{L,j}L_j(O;\hat\theta)-D_{L,j,0}L_j(O;\theta_0)\\
= \ &\widehat D_{L,j}\{L_j(O;\hat\theta)-L_j(O;\theta_0)\}+(\widehat D_{L,j}-D_{L,j,0})L_j(O;\theta_0).
\end{align*}
Thus the residual corrections at $\theta_0$ are bounded, and Lemma~\ref{lem:sm-property} gives
\begin{align*}
&\|\phi^{\beta_n}(O;\hat\theta)-\phi^{\beta_n}(O;\theta_0)\|^2\\
\lesssim{}&\|\hat e-e_0\|^2+\sum_{j=0}^{J-1}\Big\{\|\hat m_j-m_{0,j}\|^2
+\|\widehat D_{L,j}-D_{L,j,0}\|^2
+\|\widehat D_{U,j}-D_{U,j,0}\|^2\Big\}.
\end{align*}
To control the gradient differences, consider any of the four vectors $v$ in Assumption~\ref{assump:margin-cond}, and let $\hat v$ be its estimated counterpart. For fixed $t>0$, on the event $v_{(1)}-v_{(2)}>t$ and $\|\hat v-v\|_\infty<t/4$, the two vectors have the same maximizer and gaps exceeding $t/2$. As $\beta_n\to\infty$, their smoothed gradients converge uniformly on this event to the same coordinate unit vector. The complementary event has probability at most $O(t^\delta)+o_{\mathbb P}(1)$ by the margin condition and nuisance consistency. Taking $n\to\infty$ and then $t\downarrow0$, the uniform gradient bound in Lemma~\ref{lem:sm-property} yields $\|\widehat D_{L,j}-D_{L,j,0}\|+\|\widehat D_{U,j}-D_{U,j,0}\|=o_{\mathbb P}(1)$. Consequently, the correction term difference is $o_{\mathbb P}(1)$ in $L_2(\mathbb P)$, proving the required bound for Term (I).}

\medskip
\noindent\textbf{Term (II).}
{Define the empirical process by $\Gn f:=\sqrt{n}\{\Pn f-\EP[f(O)]\}$.}
Write the second term in \eqref{eq:regret-decomp-1} as ${n^{-1/2}}\Gn \Psi^{\beta_n} (O, \pi; \theta_0) $.
Notice that
\begin{equation*}
\Gn \Psi^{\beta_n} (O, \pi; {\theta_0}) = \Gn \Psi^* (O, \pi; {\theta_0}) + \left\{ \Gn \Psi^{\beta_n} (O, \pi; {\theta_0}) - \Gn \Psi^* (O, \pi; {\theta_0}) \right\},
\end{equation*}
{Here, $\Psi^*$ is the limiting score of $\Psi^{\beta_n}$ as $\beta_n\rightarrow\infty$.}
Since
\begin{equation*}
\Gn \Psi^* (O, \pi; {\theta_0}) = \sqrt{n} \left\{ \Pn \Psi^* (O, \pi; {\theta_0}) - \EP \Psi^* (O, \pi; {\theta_0}) \right\},
\end{equation*}
it suffices to prove that $ | \Gn \Psi^{\beta_n} (O, \pi; {\theta_0}) - \Gn \Psi^* (O, \pi; {\theta_0}) | =\allowbreak  o_\mathbb{P} (1) $.
Furthermore, if we can show that $ \E_\mathbb{P} | \Gn \Psi^{\beta_n} (O, \pi; {\theta_0}) - \Gn \Psi^* (O, \pi; {\theta_0}) |^2 =\allowbreak  o (1) $, then for any $ \epsilon > 0 $,
{
\begin{align*}
&\mathbb P\left(\left|\Gn\Psi^{\beta_n}(O,\pi;\theta_0)-\Gn\Psi^*(O,\pi;\theta_0)\right|\geq\epsilon\right)\\ \leq\ &\epsilon^{-2}\EP\left|\Gn\Psi^{\beta_n}(O,\pi;\theta_0)-\Gn\Psi^*(O,\pi;\theta_0)\right|^2=o(1),
\end{align*}}
which completes the proof.
In fact, we have
{
\begin{align*}
&\left\|\Gn\Psi^{\beta_n}(O,\pi;\theta_0)-\Gn\Psi^*(O,\pi;\theta_0)\right\|^2\\
={}&n^{-1}\sum_{i=1}^n\var_{\mathbb P}\left[\Psi^{\beta_n}(O_i,\pi;\theta_0)-\Psi^*(O_i,\pi;\theta_0)\right]\\
\leq{}&n^{-1}\sum_{i=1}^n\EP\left[\{\Psi^{\beta_n}(O_i,\pi;\theta_0)-\Psi^*(O_i,\pi;\theta_0)\}^2\right]\\
={}&\EP\left[\{\Psi^{\beta_n}(O,\pi;\theta_0)-\Psi^*(O,\pi;\theta_0)\}^2\right]=o(1).
\end{align*}}
Together with the bounded true outcome probabilities, the maintained bounded utilities, and Lemma~\ref{lem:sm-property}, this gives an $L_\infty(\mathbb P)$ envelope for $\Psi^{\beta_n}(O,\pi;\theta_0)$ that is uniform for all sufficiently large $n$.
The margin condition gives almost-sure convergence to $\Psi^*(O,\pi;\theta_0)$, so dominated convergence proves the final $o(1)$ bound.

\medskip
\noindent\textbf{Term (III).}
All expectations in this part of the proof are conditional on the independent sample used to estimate the nuisance functions.
Write $h_{a,j}(X)=\hat{m}_j(a,X)-m_{0,j}(a,X)$, and let $h(X)$ collect these differences over $a\in\{0,1\}$ and $j\in\mathcal J$.
By the chain rule and \eqref{eq:dr-correction},
\begin{equation*}
\phi^{\beta_n}(O;\hat{\theta})
=\sum_{a=0}^1\sum_{j=0}^{J-1}
\frac{\partial\psi^{\beta_n}(X;\hat{m})}{\partial m_j(a,X)}
T_{a,j}(O;\hat{\theta}).
\end{equation*}
Since $\EP\{T_{a,j}(O;\hat{\theta})\mid X\}=-e_0(a,X)h_{a,j}(X)/\hat{e}(a,X)$, it follows that
\begin{equation*}
\EP\{\phi^{\beta_n}(O;\hat{\theta})\mid X\}
=-\sum_{a=0}^1\sum_{j=0}^{J-1}
\frac{\partial\psi^{\beta_n}(X;\hat{m})}{\partial m_j(a,X)}
\frac{e_0(a,X)}{\hat{e}(a,X)}h_{a,j}(X).
\end{equation*}
Define the remainder
\begin{equation*}
\mathcal R_n(X)
=\psi^{\beta_n}(X;\hat{m})-\psi^{\beta_n}(X;m_0)
-\nabla_m\psi^{\beta_n}(X;\hat{m})^{\T}h(X).
\end{equation*}
The correction has conditional mean zero at $\theta_0$. Thus, by \eqref{eq:NO-score},
\begin{equation*}
\begin{aligned}
\text{(III)}
&=-\EP\Bigg[\{u^b(X)-u^n(X)\}\pi(X)\Bigg\{\mathcal R_n(X)\\
&\qquad+\sum_{a=0}^1\sum_{j=0}^{J-1}
\frac{\partial\psi^{\beta_n}(X;\hat{m})}{\partial m_j(a,X)}
\frac{\hat{e}(a,X)-e_0(a,X)}{\hat{e}(a,X)}h_{a,j}(X)
\Bigg\}\Bigg].
\end{aligned}
\end{equation*}
For fixed $J$, Lemma~\ref{lem:sm-property} and the affine dependence of the bound components on $m$ imply that $\nabla_m\psi^{\beta_n}$ is uniformly bounded and $\|\nabla_m^2\psi^{\beta_n}\|_{op}\lesssim\beta_n$.
Taylor's theorem in integral form therefore gives
\begin{equation*}
\mathcal R_n(X)
=-\int_0^1 s\,h(X)^{\T}\nabla_m^2\psi^{\beta_n}(X;m_0+s h)h(X)\,\diff s,
\end{equation*}
so that $|\mathcal R_n(X)|\lesssim\beta_n\sum_{a=0}^1\sum_{j=0}^{J-1}|h_{a,j}(X)|^2$.
Using the gradient bound, positivity of $\hat{e}$, boundedness of the utility contrast, $0\leq\pi\leq1$, and the Cauchy--Schwarz inequality, we obtain
\begin{equation*}
|\text{(III)}|
\lesssim\beta_n\sum_{a=0}^1\sum_{j=0}^{J-1}\|h_{a,j}\|^2
+\|\hat{e}-e_0\|\sum_{a=0}^1\sum_{j=0}^{J-1}\|h_{a,j}\|.
\end{equation*}
Here $\|h_{a,j}\|^2=\EP\{h_{a,j}(X)^2\}$.
By positivity of $e_0$, these treatment-specific norms are bounded by constant multiples of the corresponding norms of $\hat{m}_j(A,X)-m_{0,j}(A,X)$.
Since $J$ is fixed,
\begin{equation*}
\text{(III)}
=O_{\mathbb P}\!\left(
\|\hat{e}-e_0\|\max_{0\leq j\leq J-1}\|\hat{m}_j-m_{0,j}\|
+\beta_n\max_{0\leq j\leq J-1}\|\hat{m}_j-m_{0,j}\|^2
\right),
\end{equation*}
which establishes the required bound.

{Combining the three terms in \eqref{eq:regret-decomp-1} gives the expansion
\begin{equation*}
\begin{aligned}
&\hat R_{\sup}^{\beta_n,\no}(\pi;\hat\theta)-R_{\sup}^{\beta_n,\no}(\pi;\theta_0)\\
={}&(\Pn-\EP)\Psi^*(O,\pi;\theta_0)+o_{\mathbb P}(n^{-1/2})\\
&+O_{\mathbb P}\left(\|\hat e-e_0\|\max_j\|\hat m_j-m_{0,j}\|
+\beta_n\max_j\|\hat m_j-m_{0,j}\|^2\right).
\end{aligned}
\end{equation*}
Adding the smoothing error bounded using \eqref{eq:noone-6} yields \eqref{eq:thm1}.}

\subsection{Proof of Theorem~\ref{cor:regret-converge}}
\begin{proof}
{If $\vc(\Pi_n)=0$, all policies agree pointwise and the excess regret is zero. Otherwise, $\vc(\Pi_n)\geq1$.}
Note that the excess worst-case regret can be decomposed into
\begin{align*}
&R_{\sup} (\hat{\pi}_n; \theta_0) - R_{\sup} (\pi_n^*; \theta_0) \\ 
={}& R_{\sup} (\hat{\pi}_n; \theta_0) - \hat{R}_{\sup}^{\beta_n, \no} (\hat{\pi}_n; \hat{\theta}) + \hat{R}_{\sup}^{\beta_n, \no} (\hat{\pi}_n; \hat{\theta}) - \hat{R}_{\sup}^{\beta_n, \no} (\pi_n^*; \hat{\theta}) \\
& + \hat{R}_{\sup}^{\beta_n, \no} (\pi_n^*; \hat{\theta}) - R_{\sup} (\pi_n^*; \theta_0) \\
\leq \ & 2 \sup_{\pi \in \Pi_n} | \hat{R}_{\sup}^{\beta_n, \no} (\pi; \hat{\theta}) - R_{\sup} (\pi; \theta_0) |,
\end{align*}
where we have used the fact that $ \hat{\pi}_n $ minimizes $ \hat{R}_{\sup}^{\beta_n, \no} (\pi; \hat{\theta}) $.

{Condition on the independent sample used to estimate $\hat\theta$. The score class $\{\Psi^{\beta_n}(\cdot,\pi;\hat\theta):\pi\in\Pi_n\}$ consists of binary policies multiplied by a fixed, uniformly bounded function, by the boundedness of $(\hat e,\hat m)$ and Lemma~\ref{lem:sm-property}. The conditional VC maximal inequality \citep{wainwright2019high} therefore gives
\[
\sup_{\pi\in\Pi_n}|(\Pn-\EP)\Psi^{\beta_n}(O,\pi;\hat\theta)|
=O_{\mathbb P}\!\left(\sqrt{\vc(\Pi_n)/n}\right).
\]
The smoothing error bound and the bound for Term (III) established in the proof of Theorem~\ref{thm:asymp-converg} hold uniformly in $\pi$, since $|\pi|\leq1$. Decomposing the estimation error into its centered empirical part and population bias gives
\begin{equation*}
\begin{aligned}
&\sup_{\pi\in\Pi_n}|\hat R_{\sup}^{\beta_n,\no}(\pi;\hat\theta)-R_{\sup}(\pi;\theta_0)|\\
={}&O_{\mathbb P}\!\left(\sqrt{\vc(\Pi_n)/n}+\beta_n^{-(1+\delta)}\right)\\
&+O_{\mathbb P}\!\left(\|\hat e-e_0\|\max_j\|\hat m_j-m_{0,j}\|
+\beta_n\max_j\|\hat m_j-m_{0,j}\|^2\right).
\end{aligned}
\end{equation*}
Combining this uniform bound with the preceding excess regret inequality proves the theorem.}
\end{proof}

\section{Simulation Details and Additional Results}\label{app:add-sim}

Two baseline covariates $ X = (X_1, X_2) $ are generated independently from $ \text{Unif} (-1, 1) \times \text{Unif} (-1, 1) $, and $ A \mid X \sim \text{Bernoulli} (e(X)) $, where $ e(X) = \min \{ \max \{ X_1^2, 0.1 \}, 0.9 \} $.
This data-generating process follows a similar setup to \cite{levis2025covariate}.
The ordinal outcome is generated by a multinomial logistic distribution specified as
\begin{equation*}
\begin{gathered} m_j (a, X) = \prob (Y_a = j \mid X) = \frac{ \exp ( \alpha_{a,j}^\T \tilde{X} ) }{ 1+\sum_{l = 1}^{J - 1} \exp(\alpha_{a,l}^\T \tilde{X}) }, 
\end{gathered}
\end{equation*}
for $a = 0, 1, ~ j = 0, 1, \ldots, J - 1,$ where $ \tilde{X} = (1, X^{\T})^{\T} $ is the covariate vector including an intercept, and elements in $ \alpha_{a,j} $ are constants in $ (-1, 1) $.
Let $ \alpha_{a,0} = 0 $ for $a=0,1$, with $j = 0$ set as the reference level.
The coefficients $ \alpha_{1,j} $, $j = 1, \ldots, 7$ for $ Y_1 $ are specified as
\begin{equation*}
\begin{aligned}
& \alpha_{1,0} = \begin{pmatrix} 0 \\ 0 \\ 0 \end{pmatrix}, ~~ \alpha_{1,1} = \begin{pmatrix} -0.4 \\ -0.3 \\ -0.1 \end{pmatrix}, ~~ \alpha_{1,2} = \begin{pmatrix} -0.35 \\ 0.9 \\ -0.6 \end{pmatrix}, ~~ \alpha_{1,3} = \begin{pmatrix} -0.3 \\ -0.7 \\ 0.2 \end{pmatrix}, \\
& \alpha_{1,4} = \begin{pmatrix} -0.25 \\ 0.7 \\ 0.5 \end{pmatrix}, ~~ \alpha_{1,5} = \begin{pmatrix} -0.2 \\ 0.3 \\ 0.7 \end{pmatrix}, ~~ \alpha_{1,6} = \begin{pmatrix} -0.15 \\ -0.4 \\ -0.1 \end{pmatrix}, ~~ \alpha_{1,7} = \begin{pmatrix} -0.1 \\ 0 \\ 0.2 \end{pmatrix}.
\end{aligned}
\end{equation*}
{For each $J\in\{3,5,8\}$, we use the coefficient vectors $\alpha_{1,j}$ with $j=0,\ldots,J-1$ from the displayed list.}
For $ Y_0 $, each coefficient vector $\alpha_{0,j}$ is identical to $\alpha_{1,j}$, except that the third element (i.e., the coefficient of \(X_2\)) takes the opposite sign.
For simplicity, we take the utility threshold defined in \eqref{eq:asym-shift} to be constant throughout the simulation study, i.e., $ C_u (x) \equiv C_u $ for all $ x \in \mathcal{X} $.

Figure~\ref{fig:rsup-j5} shows that, for $J=5$, the direct influence function and orthogonal smoothed estimators closely track each other and sharply reduce RMSE as the nuisance rate increases; they also tend to attain lower excess regret at moderate rates, although the plug-in estimators can remain competitive at the fastest rates.

\begin{figure}[htbp]
\centering
\includegraphics[width=0.92\textwidth]{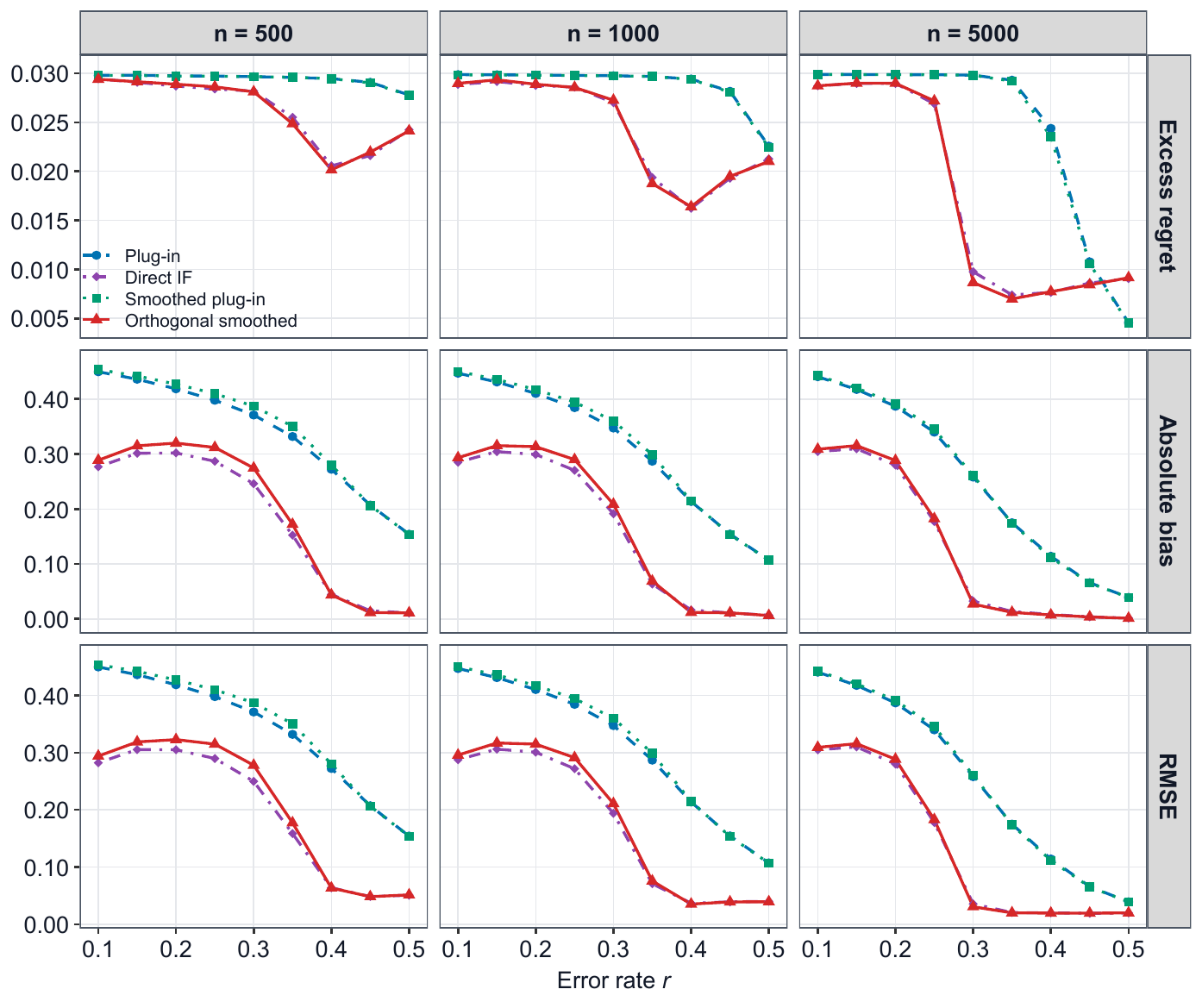}
\caption{Excess worst-case regret, absolute bias, and RMSE of regret estimation over 500 replications, under different sample sizes and error rates. Here $J=5$.}
\label{fig:rsup-j5}
\end{figure}

Figure~\ref{fig:rsup-j8} shows a clearer distinction between regret estimation and policy learning when $J=8$: the estimators with influence function corrections substantially reduce RMSE once the nuisance rate is moderate, whereas the plug-in estimators generally achieve lower excess regret.
Thus, improved accuracy in estimating worst-case regret does not necessarily translate into lower finite-sample excess regret for the learned policy.

\begin{figure}[htbp]
\centering
\includegraphics[width=0.92\textwidth]{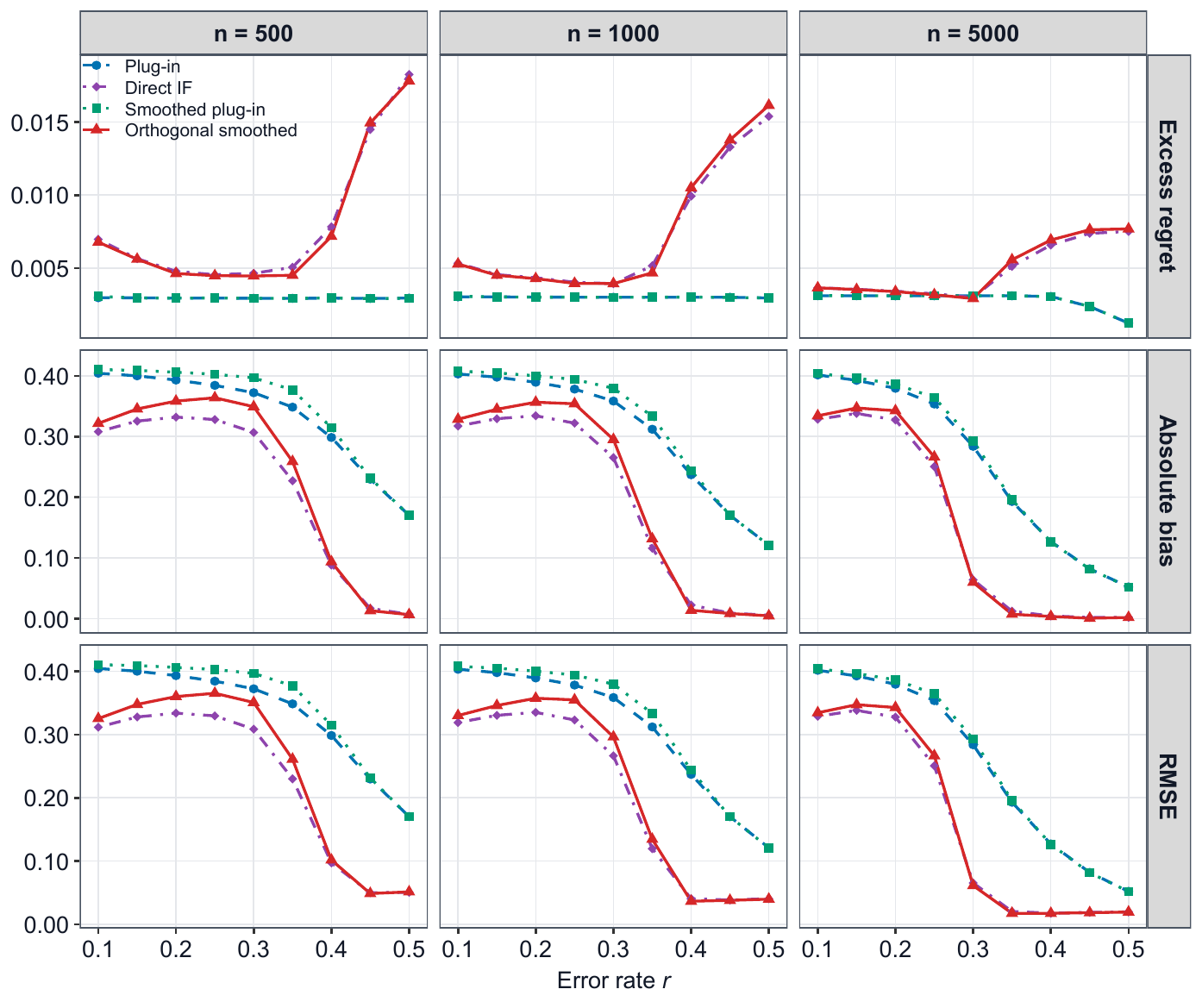}
\caption{Excess worst-case regret, absolute bias, and RMSE of regret estimation over 500 replications, under different sample sizes and error rates. Here $J=8$.}
\label{fig:rsup-j8}
\end{figure}

Tables~\ref{tab:rsup-j3}, \ref{tab:rsup-j5}, and~\ref{tab:rsup-j8} report the numerical results for $J=3$, 5, and 8, respectively. All three measures are multiplied by 100; absolute bias is the magnitude of the average signed estimation error, not the mean absolute error.

\begin{table}[p]
\renewcommand{\baselinestretch}{1}\small
\centering
\caption{{Absolute bias, RMSE, and mean excess worst-case regret ($\times 100$) for $J=3$ over 500 replications.}}
\label{tab:rsup-j3}
\begingroup
\scriptsize
\setlength{\tabcolsep}{1.4pt}
\renewcommand{\arraystretch}{1.12}
\begin{tabular}{@{}cr*{12}{r}@{}}
\toprule
& & \multicolumn{3}{c}{Direct plug-in} & \multicolumn{3}{c}{Direct IF} & \multicolumn{3}{c}{Smoothed plug-in} & \multicolumn{3}{c}{Orthogonal smoothed} \\
\cmidrule(lr){3-5}\cmidrule(lr){6-8}\cmidrule(lr){9-11}\cmidrule(lr){12-14}
$n$ & $r$ & Abs. bias & RMSE & Excess & Abs. bias & RMSE & Excess & Abs. bias & RMSE & Excess & Abs. bias & RMSE & Excess \\
\midrule
\multirow[c]{9}{*}{500} & 0.10 & 52.938 & 52.946 & 11.093 & 25.503 & 26.163 & 8.692 & 53.180 & 53.188 & 11.074 & 26.477 & 27.151 & 8.728 \\
 & 0.15 & 50.372 & 50.379 & 11.082 & 28.841 & 29.315 & 9.159 & 50.806 & 50.813 & 11.069 & 30.248 & 30.678 & 9.229 \\
 & 0.20 & 47.365 & 47.373 & 11.050 & 29.293 & 29.666 & 9.374 & 48.060 & 48.067 & 11.043 & 31.180 & 31.525 & 9.481 \\
 & 0.25 & 43.979 & 43.985 & 11.009 & 25.476 & 25.909 & 9.020 & 44.998 & 45.005 & 11.000 & 27.716 & 28.115 & 9.171 \\
 & 0.30 & 39.490 & 39.498 & 10.896 & 18.227 & 18.820 & 8.067 & 40.783 & 40.791 & 10.885 & 20.261 & 20.825 & 8.264 \\
 & 0.35 & 33.315 & 33.325 & 10.514 & 7.168 & 8.691 & 4.849 & 34.590 & 34.600 & 10.474 & 8.065 & 9.433 & 4.967 \\
 & 0.40 & 25.355 & 25.368 & 9.467 & 1.067 & 4.119 & 2.612 & 25.944 & 25.958 & 9.435 & 0.941 & 3.996 & 2.424 \\
 & 0.45 & 17.505 & 17.522 & 7.878 & 0.103 & 4.303 & 2.257 & 17.561 & 17.577 & 7.927 & 0.177 & 4.158 & 2.216 \\
 & 0.50 & 10.363 & 10.396 & 4.972 & 0.216 & 4.635 & 2.169 & 10.372 & 10.405 & 5.035 & 0.127 & 4.632 & 2.191 \\
\midrule
\multirow[c]{9}{*}{1000} & 0.10 & 52.554 & 52.558 & 11.136 & 26.881 & 27.189 & 8.901 & 52.755 & 52.759 & 11.124 & 27.707 & 28.008 & 8.940 \\
 & 0.15 & 49.516 & 49.520 & 11.126 & 30.190 & 30.410 & 9.505 & 49.891 & 49.894 & 11.121 & 31.292 & 31.509 & 9.536 \\
 & 0.20 & 46.068 & 46.072 & 11.102 & 28.660 & 28.855 & 9.559 & 46.686 & 46.690 & 11.097 & 30.198 & 30.380 & 9.657 \\
 & 0.25 & 41.734 & 41.737 & 11.053 & 22.965 & 23.230 & 8.974 & 42.647 & 42.651 & 11.048 & 24.976 & 25.217 & 9.186 \\
 & 0.30 & 35.749 & 35.753 & 10.785 & 10.763 & 11.449 & 5.893 & 36.782 & 36.786 & 10.763 & 11.939 & 12.551 & 6.288 \\
 & 0.35 & 27.232 & 27.238 & 9.796 & 1.574 & 3.417 & 1.950 & 28.061 & 28.067 & 9.751 & 1.697 & 3.426 & 1.800 \\
 & 0.40 & 18.405 & 18.414 & 8.157 & 0.105 & 3.051 & 1.442 & 18.606 & 18.615 & 8.172 & 0.036 & 3.110 & 1.319 \\
 & 0.45 & 10.352 & 10.371 & 4.980 & 0.168 & 3.200 & 1.376 & 10.439 & 10.458 & 5.153 & 0.122 & 3.258 & 1.368 \\
 & 0.50 & 4.532 & 4.554 & 1.735 & 0.233 & 3.392 & 1.381 & 4.463 & 4.487 & 1.722 & 0.223 & 3.455 & 1.363 \\
\midrule
\multirow[c]{9}{*}{5000} & 0.10 & 51.221 & 51.222 & 11.140 & 29.636 & 29.689 & 9.268 & 51.351 & 51.352 & 11.139 & 30.076 & 30.128 & 9.269 \\
 & 0.15 & 47.224 & 47.225 & 11.135 & 30.426 & 30.477 & 9.779 & 47.473 & 47.473 & 11.133 & 31.117 & 31.166 & 9.827 \\
 & 0.20 & 42.191 & 42.192 & 11.112 & 24.276 & 24.338 & 9.260 & 42.629 & 42.629 & 11.111 & 25.240 & 25.299 & 9.335 \\
 & 0.25 & 34.597 & 34.598 & 10.758 & 8.810 & 9.081 & 4.495 & 35.125 & 35.126 & 10.752 & 9.555 & 9.787 & 4.945 \\
 & 0.30 & 23.631 & 23.632 & 9.226 & 0.230 & 1.270 & 0.663 & 23.891 & 23.893 & 9.198 & 0.118 & 1.248 & 0.620 \\
 & 0.35 & 13.133 & 13.138 & 6.252 & 0.174 & 1.384 & 0.616 & 13.442 & 13.447 & 6.546 & 0.298 & 1.369 & 0.605 \\
 & 0.40 & 5.160 & 5.169 & 2.100 & 0.149 & 1.445 & 0.633 & 4.950 & 4.962 & 2.055 & 0.219 & 1.448 & 0.640 \\
 & 0.45 & 2.011 & 2.013 & 0.396 & 0.086 & 1.497 & 0.666 & 1.921 & 1.923 & 0.349 & 0.102 & 1.492 & 0.649 \\
 & 0.50 & 1.081 & 1.081 & 0.105 & 0.029 & 1.538 & 0.733 & 1.073 & 1.073 & 0.102 & 0.037 & 1.512 & 0.736 \\
\bottomrule
\end{tabular}
\endgroup
\end{table}

\begin{table}[p]
\renewcommand{\baselinestretch}{1}\small
\centering
\caption{{Absolute bias, RMSE, and mean excess worst-case regret ($\times 100$) for $J=5$ over 500 replications.}}
\label{tab:rsup-j5}
\begingroup
\scriptsize
\setlength{\tabcolsep}{1.4pt}
\renewcommand{\arraystretch}{1.12}
\begin{tabular}{@{}cr*{12}{r}@{}}
\toprule
& & \multicolumn{3}{c}{Direct plug-in} & \multicolumn{3}{c}{Direct IF} & \multicolumn{3}{c}{Smoothed plug-in} & \multicolumn{3}{c}{Orthogonal smoothed} \\
\cmidrule(lr){3-5}\cmidrule(lr){6-8}\cmidrule(lr){9-11}\cmidrule(lr){12-14}
$n$ & $r$ & Abs. bias & RMSE & Excess & Abs. bias & RMSE & Excess & Abs. bias & RMSE & Excess & Abs. bias & RMSE & Excess \\
\midrule
\multirow[c]{9}{*}{500} & 0.10 & 44.964 & 44.971 & 2.979 & 27.676 & 28.222 & 2.943 & 45.361 & 45.368 & 2.977 & 28.868 & 29.390 & 2.938 \\
 & 0.15 & 43.606 & 43.612 & 2.980 & 30.132 & 30.522 & 2.908 & 44.202 & 44.208 & 2.979 & 31.491 & 31.865 & 2.914 \\
 & 0.20 & 41.853 & 41.859 & 2.973 & 30.191 & 30.496 & 2.869 & 42.704 & 42.709 & 2.971 & 31.983 & 32.272 & 2.888 \\
 & 0.25 & 39.779 & 39.783 & 2.971 & 28.682 & 28.971 & 2.840 & 40.980 & 40.984 & 2.971 & 31.195 & 31.481 & 2.861 \\
 & 0.30 & 37.095 & 37.099 & 2.965 & 24.596 & 24.973 & 2.818 & 38.705 & 38.709 & 2.965 & 27.451 & 27.788 & 2.811 \\
 & 0.35 & 33.184 & 33.188 & 2.958 & 15.242 & 15.852 & 2.551 & 35.036 & 35.040 & 2.960 & 17.205 & 17.728 & 2.485 \\
 & 0.40 & 27.226 & 27.231 & 2.945 & 4.463 & 6.377 & 2.056 & 28.002 & 28.007 & 2.945 & 4.374 & 6.336 & 2.015 \\
 & 0.45 & 20.706 & 20.709 & 2.907 & 1.489 & 4.840 & 2.161 & 20.624 & 20.627 & 2.902 & 1.148 & 4.825 & 2.195 \\
 & 0.50 & 15.413 & 15.416 & 2.779 & 1.075 & 5.078 & 2.424 & 15.366 & 15.369 & 2.776 & 1.109 & 5.132 & 2.413 \\
\midrule
\multirow[c]{9}{*}{1000} & 0.10 & 44.695 & 44.698 & 2.986 & 28.511 & 28.763 & 2.886 & 45.008 & 45.012 & 2.982 & 29.321 & 29.568 & 2.897 \\
 & 0.15 & 43.100 & 43.103 & 2.985 & 30.440 & 30.601 & 2.915 & 43.577 & 43.580 & 2.984 & 31.526 & 31.680 & 2.933 \\
 & 0.20 & 41.015 & 41.018 & 2.981 & 29.912 & 30.064 & 2.876 & 41.717 & 41.719 & 2.981 & 31.333 & 31.480 & 2.888 \\
 & 0.25 & 38.435 & 38.437 & 2.979 & 27.032 & 27.182 & 2.859 & 39.445 & 39.447 & 2.977 & 28.984 & 29.122 & 2.855 \\
 & 0.30 & 34.724 & 34.726 & 2.975 & 19.133 & 19.384 & 2.700 & 36.011 & 36.013 & 2.974 & 20.865 & 21.091 & 2.725 \\
 & 0.35 & 28.684 & 28.686 & 2.967 & 6.351 & 7.080 & 1.936 & 29.880 & 29.883 & 2.967 & 6.886 & 7.565 & 1.875 \\
 & 0.40 & 21.363 & 21.365 & 2.941 & 1.636 & 3.676 & 1.622 & 21.459 & 21.461 & 2.938 & 1.186 & 3.526 & 1.639 \\
 & 0.45 & 15.435 & 15.437 & 2.816 & 1.081 & 3.911 & 1.929 & 15.326 & 15.327 & 2.807 & 1.111 & 3.909 & 1.948 \\
 & 0.50 & 10.718 & 10.722 & 2.257 & 0.677 & 4.016 & 2.127 & 10.680 & 10.684 & 2.243 & 0.620 & 3.930 & 2.101 \\
\midrule
\multirow[c]{9}{*}{5000} & 0.10 & 44.058 & 44.058 & 2.988 & 30.448 & 30.489 & 2.872 & 44.235 & 44.236 & 2.987 & 30.872 & 30.913 & 2.872 \\
 & 0.15 & 41.742 & 41.743 & 2.987 & 30.958 & 30.989 & 2.896 & 42.018 & 42.019 & 2.987 & 31.529 & 31.561 & 2.899 \\
 & 0.20 & 38.687 & 38.687 & 2.986 & 27.985 & 28.014 & 2.896 & 39.113 & 39.113 & 2.986 & 28.818 & 28.847 & 2.897 \\
 & 0.25 & 33.977 & 33.977 & 2.985 & 17.712 & 17.769 & 2.689 & 34.540 & 34.541 & 2.985 & 18.239 & 18.294 & 2.718 \\
 & 0.30 & 25.796 & 25.797 & 2.981 & 3.240 & 3.566 & 0.976 & 26.071 & 26.072 & 2.980 & 2.693 & 3.060 & 0.865 \\
 & 0.35 & 17.495 & 17.496 & 2.930 & 1.346 & 2.052 & 0.739 & 17.374 & 17.375 & 2.919 & 1.195 & 1.997 & 0.699 \\
 & 0.40 & 11.404 & 11.405 & 2.438 & 0.816 & 1.998 & 0.764 & 11.158 & 11.160 & 2.353 & 0.720 & 1.955 & 0.773 \\
 & 0.45 & 6.589 & 6.590 & 1.076 & 0.379 & 1.914 & 0.861 & 6.528 & 6.529 & 1.055 & 0.360 & 1.943 & 0.842 \\
 & 0.50 & 3.917 & 3.918 & 0.455 & 0.144 & 1.980 & 0.906 & 3.908 & 3.909 & 0.451 & 0.129 & 1.985 & 0.915 \\
\bottomrule
\end{tabular}
\endgroup
\end{table}

\begin{table}[p]
\renewcommand{\baselinestretch}{1}\small
\centering
\caption{{Absolute bias, RMSE, and mean excess worst-case regret ($\times 100$) for $J=8$ over 500 replications.}}
\label{tab:rsup-j8}
\begingroup
\scriptsize
\setlength{\tabcolsep}{1.4pt}
\renewcommand{\arraystretch}{1.12}
\begin{tabular}{@{}cr*{12}{r}@{}}
\toprule
& & \multicolumn{3}{c}{Direct plug-in} & \multicolumn{3}{c}{Direct IF} & \multicolumn{3}{c}{Smoothed plug-in} & \multicolumn{3}{c}{Orthogonal smoothed} \\
\cmidrule(lr){3-5}\cmidrule(lr){6-8}\cmidrule(lr){9-11}\cmidrule(lr){12-14}
$n$ & $r$ & Abs. bias & RMSE & Excess & Abs. bias & RMSE & Excess & Abs. bias & RMSE & Excess & Abs. bias & RMSE & Excess \\
\midrule
\multirow[c]{9}{*}{500} & 0.10 & 40.422 & 40.428 & 0.297 & 30.793 & 31.160 & 0.697 & 41.031 & 41.037 & 0.305 & 32.164 & 32.517 & 0.677 \\
 & 0.15 & 39.998 & 40.003 & 0.296 & 32.534 & 32.771 & 0.564 & 40.898 & 40.904 & 0.295 & 34.540 & 34.759 & 0.561 \\
 & 0.20 & 39.305 & 39.309 & 0.294 & 33.201 & 33.381 & 0.478 & 40.596 & 40.601 & 0.294 & 35.834 & 35.997 & 0.462 \\
 & 0.25 & 38.427 & 38.432 & 0.294 & 32.797 & 32.945 & 0.454 & 40.261 & 40.265 & 0.295 & 36.381 & 36.516 & 0.448 \\
 & 0.30 & 37.217 & 37.221 & 0.293 & 30.665 & 30.848 & 0.464 & 39.668 & 39.672 & 0.293 & 34.914 & 35.064 & 0.445 \\
 & 0.35 & 34.838 & 34.842 & 0.292 & 22.713 & 22.996 & 0.505 & 37.675 & 37.680 & 0.293 & 25.883 & 26.108 & 0.450 \\
 & 0.40 & 29.845 & 29.849 & 0.294 & 8.821 & 9.721 & 0.783 & 31.474 & 31.478 & 0.294 & 9.361 & 10.173 & 0.717 \\
 & 0.45 & 22.966 & 22.969 & 0.292 & 1.691 & 5.007 & 1.450 & 23.148 & 23.151 & 0.291 & 1.306 & 4.880 & 1.494 \\
 & 0.50 & 17.032 & 17.035 & 0.294 & 0.701 & 5.063 & 1.827 & 17.032 & 17.034 & 0.292 & 0.655 & 5.134 & 1.782 \\
\midrule
\multirow[c]{9}{*}{1000} & 0.10 & 40.306 & 40.309 & 0.302 & 31.735 & 31.896 & 0.533 & 40.771 & 40.773 & 0.306 & 32.841 & 33.000 & 0.528 \\
 & 0.15 & 39.777 & 39.779 & 0.302 & 32.936 & 33.035 & 0.455 & 40.484 & 40.487 & 0.302 & 34.491 & 34.582 & 0.450 \\
 & 0.20 & 38.936 & 38.939 & 0.300 & 33.419 & 33.503 & 0.433 & 40.006 & 40.008 & 0.300 & 35.651 & 35.732 & 0.428 \\
 & 0.25 & 37.814 & 37.816 & 0.300 & 32.199 & 32.293 & 0.403 & 39.401 & 39.403 & 0.300 & 35.392 & 35.473 & 0.395 \\
 & 0.30 & 35.841 & 35.843 & 0.300 & 26.524 & 26.640 & 0.393 & 37.948 & 37.950 & 0.299 & 29.516 & 29.623 & 0.393 \\
 & 0.35 & 31.197 & 31.199 & 0.299 & 11.590 & 11.966 & 0.515 & 33.332 & 33.335 & 0.300 & 13.153 & 13.447 & 0.467 \\
 & 0.40 & 23.703 & 23.705 & 0.300 & 2.230 & 4.058 & 0.993 & 24.372 & 24.373 & 0.301 & 1.374 & 3.659 & 1.051 \\
 & 0.45 & 17.020 & 17.022 & 0.300 & 0.907 & 3.851 & 1.329 & 17.019 & 17.020 & 0.299 & 0.836 & 3.799 & 1.378 \\
 & 0.50 & 12.096 & 12.097 & 0.295 & 0.556 & 4.065 & 1.539 & 12.077 & 12.079 & 0.295 & 0.458 & 3.989 & 1.615 \\
\midrule
\multirow[c]{9}{*}{5000} & 0.10 & 40.149 & 40.149 & 0.310 & 32.844 & 32.870 & 0.365 & 40.392 & 40.393 & 0.311 & 33.413 & 33.439 & 0.365 \\
 & 0.15 & 39.254 & 39.254 & 0.310 & 33.801 & 33.818 & 0.354 & 39.654 & 39.654 & 0.310 & 34.700 & 34.716 & 0.353 \\
 & 0.20 & 37.945 & 37.946 & 0.310 & 32.736 & 32.750 & 0.342 & 38.634 & 38.635 & 0.310 & 34.251 & 34.266 & 0.338 \\
 & 0.25 & 35.359 & 35.359 & 0.310 & 25.028 & 25.053 & 0.323 & 36.432 & 36.432 & 0.310 & 26.632 & 26.654 & 0.317 \\
 & 0.30 & 28.384 & 28.385 & 0.310 & 6.392 & 6.547 & 0.298 & 29.278 & 29.279 & 0.310 & 5.974 & 6.136 & 0.291 \\
 & 0.35 & 19.301 & 19.301 & 0.310 & 1.217 & 2.046 & 0.511 & 19.581 & 19.581 & 0.310 & 0.719 & 1.751 & 0.556 \\
 & 0.40 & 12.706 & 12.707 & 0.305 & 0.439 & 1.772 & 0.658 & 12.643 & 12.643 & 0.304 & 0.359 & 1.750 & 0.692 \\
 & 0.45 & 8.205 & 8.205 & 0.236 & 0.187 & 1.893 & 0.738 & 8.154 & 8.154 & 0.237 & 0.063 & 1.854 & 0.763 \\
 & 0.50 & 5.214 & 5.214 & 0.122 & 0.206 & 1.977 & 0.751 & 5.198 & 5.198 & 0.123 & 0.189 & 1.950 & 0.768 \\
\bottomrule
\end{tabular}
\endgroup
\end{table}

Additional sensitivity results for the treated proportion at the smaller and larger sample sizes are displayed in Figures~\ref{fig:trpr-n500} and~\ref{fig:trpr-n5000}, respectively.

\begin{figure}[htbp]
\centering
\includegraphics[width=0.96\textwidth]{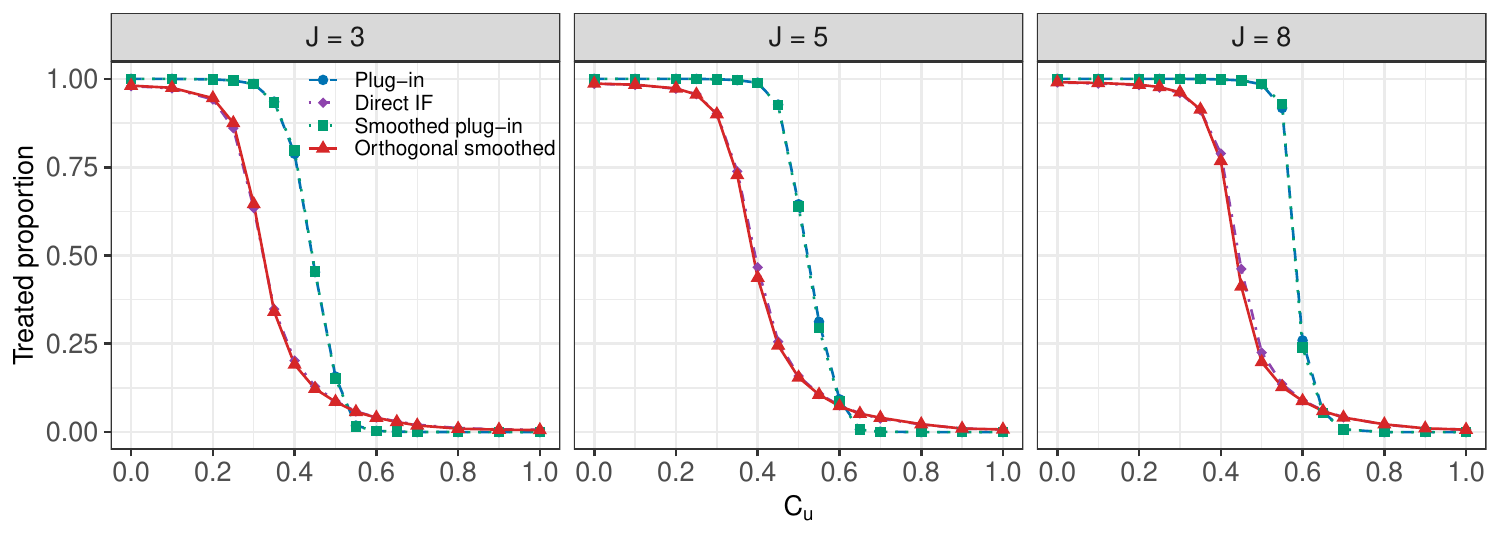}
\caption{Mean treated proportion over 500 replications as the utility threshold $C_u$ varies, with sample size $n=500$.
Columns correspond to $J=3$, 5, and 8 outcome levels.}
\label{fig:trpr-n500}
\end{figure}

\begin{figure}[htbp]
\centering
\includegraphics[width=0.96\textwidth]{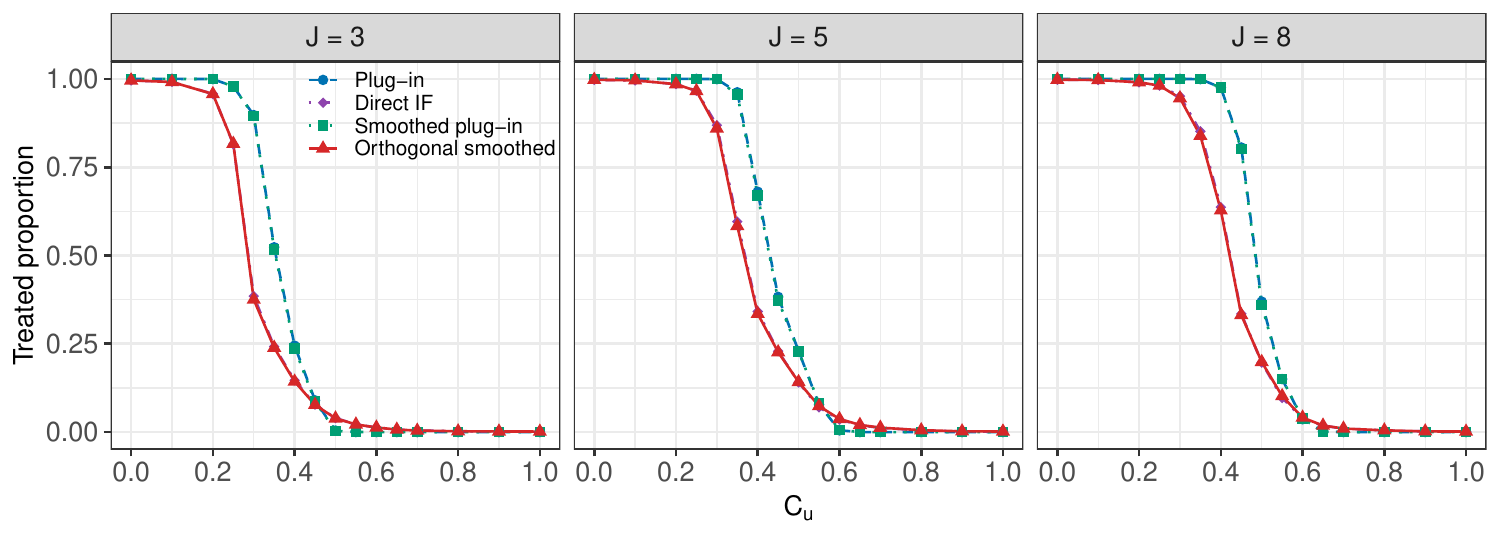}
\caption{Mean treated proportion over 500 replications as the utility threshold $C_u$ varies, with sample size $n=5000$.
Columns correspond to $J=3$, 5, and 8 outcome levels.}
\label{fig:trpr-n5000}
\end{figure}

\clearpage
\section{Additional SIPP results}\label{app:sipp-results}

For the nuisance models in Section~\ref{sec:app}, age and income-to-poverty ratio are standardized before resampling. Education is coded from 1 to 4 and enters linearly; gender, Black race, Asian race, Hispanic ethnicity, marital status, insurance coverage, and work limitation enter as binary indicators. Outcome probabilities $\hat m_j(a,X)$ are predicted with treatment set to $a=0$ and $a=1$, and propensity predictions are truncated to $[0.05,0.95]$.

In each bootstrap replication, individuals are randomly partitioned into three folds of 2773, stratified by treatment and outcome so that each stratum's fold counts differ by at most one. We then draw 2773 records with replacement within each fold, keeping all copies of an individual together. The folds rotate cyclically through nuisance estimation, policy learning, and evaluation, with each fold serving each role once. The partition is redrawn in every replication, while fold assignments and fitted nuisance models are shared across methods and utility thresholds within a replication. Rotation-specific treatment proportions and worst-case regret estimates are averaged equally before computing bootstrap summaries.

The utility-threshold grid combines $0.1$-$0.6$ in increments of $0.05$ with $0.24$-$0.34$ in increments of $0.01$, giving 20 distinct values. We normalize $u^b-u^n=1$, with $u^b=1-C_u$ and $u^n=-C_u$. For the smoothed methods, $\beta_n=2n_{\mathrm{score}}^{1/4}\approx14.51$, where $n_{\mathrm{score}}=2773$ is the number of observations used to compute the scores.

Table~\ref{tab:sipp-Rsup-new} reports treated proportions (Panel A) and estimated worst-case regret (Panel B) at selected thresholds. Regret is evaluated on held-out samples using the corresponding estimators in Table~\ref{tab:estimators}. The common policy-independent additive constant is omitted as in Section~\ref{sec:smooth-approx}, so negative entries do not indicate negative regret.

\begin{table}[htbp]
\renewcommand{\baselinestretch}{1}\small
\centering
\caption{Bootstrap means (SDs) of the SIPP application at selected utility thresholds over 100 bootstrap replications}
\label{tab:sipp-Rsup-new}
\small
\setlength{\tabcolsep}{5pt}
\begin{tabular}{@{}crrrr@{}}
\toprule
$C_u$ & Direct plug-in & Direct IF & Smoothed plug-in & Orthogonal smoothed \\
\midrule
\multicolumn{5}{c}{A: Treated proportion} \\
\midrule
0.10 & {0.999 (0.001)} & {0.962 (0.014)} & {1.000 (0.001)} & {0.977 (0.007)} \\
0.20 & {0.997 (0.002)} & {0.883 (0.043)} & {0.997 (0.002)} & {0.916 (0.027)} \\
0.25 & {0.966 (0.023)} & {0.718 (0.072)} & {0.981 (0.014)} & {0.799 (0.061)} \\
0.26 & {0.921 (0.057)} & {0.667 (0.086)} & 0.962 (0.033) & {0.749 (0.074)} \\
0.27 & {0.819 (0.108)} & {0.605 (0.089)} & {0.919 (0.063)} & {0.695 (0.085)} \\
0.28 & {0.654 (0.162)} & {0.530 (0.086)} & {0.821 (0.117)} & {0.641 (0.092)} \\
0.29 & {0.455 (0.184)} & {0.464 (0.088)} & {0.660 (0.167)} & {0.576 (0.094)} \\
0.30 & {0.259 (0.161)} & {0.394 (0.089)} & {0.461 (0.184)} & {0.497 (0.088)} \\
0.35 & {0.012 (0.028)} & {0.172 (0.063)} & {0.015 (0.031)} & {0.202 (0.058)} \\
0.60 & 0.000 (0.000) & {0.008 (0.005)} & 0.000 (0.000) & {0.005 (0.003)} \\
\midrule
\multicolumn{5}{c}{B: Estimated worst-case regret} \\
\midrule
0.10 & {-0.376 (0.018)} & {-0.368 (0.018)} & {-0.438 (0.017)} & {-0.448 (0.018)} \\
0.20 & {-0.176 (0.018)} & {-0.161 (0.021)} & {-0.224 (0.019)} & {-0.220 (0.021)} \\
0.25 & {-0.077 (0.017)} & {-0.058 (0.020)} & {-0.109 (0.019)} & {-0.100 (0.024)} \\
0.26 & {-0.058 (0.017)} & {-0.041 (0.020)} & {-0.087 (0.019)} & {-0.076 (0.023)} \\
0.27 & {-0.041 (0.016)} & {-0.024 (0.018)} & {-0.065 (0.018)} & {-0.054 (0.023)} \\
0.28 & {-0.026 (0.014)} & {-0.010 (0.014)} & {-0.045 (0.017)} & {-0.035 (0.021)} \\
0.29 & {-0.015 (0.011)} & {-0.002 (0.013)} & {-0.028 (0.014)} & {-0.019 (0.017)} \\
0.30 & {-0.008 (0.008)} & {0.005 (0.013)} & {-0.016 (0.011)} & {-0.005 (0.016)} \\
0.35 & {-0.000 (0.001)} & {0.017 (0.012)} & {-0.000 (0.001)} & {0.016 (0.012)} \\
0.60 & 0.000 (0.000) & {0.005 (0.004)} & 0.000 (0.000) & {0.003 (0.002)} \\
\bottomrule
\end{tabular}
\end{table}

\end{appendices}

\end{document}